\documentclass[11pt]{article} 

\usepackage[utf8]{inputenc}
\usepackage[T1]{fontenc}

\usepackage{amsmath}
\usepackage{amsthm}
\usepackage{amssymb}
\usepackage{mathtools}

\usepackage{mathbbol} 

\usepackage{mathrsfs} 

\usepackage{natbib}

\usepackage[dvipsnames]{xcolor}
\usepackage{colortbl} 
\usepackage{graphicx}

\usepackage{booktabs}
\usepackage{multirow}
\usepackage{caption}
\usepackage{subcaption}

\usepackage{enumitem} 

\usepackage{geometry}
\usepackage{changepage} 

\usepackage[hidelinks]{hyperref} 
\usepackage{url}

\usepackage{mdframed}

\usepackage{setspace}

\usepackage{authblk}

\usepackage[ruled,vlined]{algorithm2e} 
\DontPrintSemicolon
\SetKwInOut{Input}{Input}
\SetKwInOut{Output}{Output}

\usepackage{centernot}
\usepackage{multicol}
\usepackage{adjustbox}

\newcommand{\R}{\mathbb{R}}
\newcommand{\E}{\mathbb{E}}
\newcommand{\Prob}{\Pr}

\newcommand{\Var}{\mathrm{Var}}
\newcommand{\N} {\mathcal{N}}
\newcommand{\CI} {\mathrm{CI}}

\newcommand{\dx}[1]{\ \mathrm{d} #1}
\newcommand{\dy}[1]{\text{d}#1}
\newcommand{\ff}[2]{\frac{#1}{#2}}
\newcommand{\dd}[2]{\ff{\dx{#1}}{\dx{#2}}}
\newcommand{\pp}[2]{\ff{\partial{#1}}{\partial{#2}}}
\newcommand{\pd}[1]{\partial_{#1}}

\newcommand{\indep}{\rotatebox[origin=c]{90}{$\models$}}

\newcommand{\I}[1]{\mathbb{1}\!\left\{ #1 \right\}}
\newcommand{\Ind}{\mathbb{1}}

\newcommand{\lam}{\lambda}
\newcommand{\IF}[1]{\mathbb{IF}[#1]}
\newcommand{\EIF}{\mathbb{IF}}

\newcommand{\bmidl}{\bigg|_{\lambda = 0}}
\newcommand{\midl}{|_{\lambda = 0}}
\newcommand{\tgamma}{\tilde{\gamma}}
\newcommand{\plam}{\partial_{\lam}}
\newcommand{\dequiv}{\hspace{0.5em} \overset{d}{=} \hspace{0.5em}}
\newcommand{\supp}{\mathrm{Supp}}
\newcommand{\score}{\mathbb{S}}
\newcommand{\med}{\mathrm{Median}}

\newcommand{\pto}{\overset{P}{\to}}
\newcommand{\dto}{\overset{\mathcal{D}}{\to}}
\newcommand{\expit}{\operatorname{expit}}

\newtheorem{theorem}{Theorem}
\newtheorem{proposition}{Proposition}
\newtheorem{lemma}{Lemma}
\newtheorem{assumption}{Assumption}
\newtheorem{remark}{Remark}
\newtheorem{corollary}{Corollary}

\theoremstyle{definition}
\newtheorem{example}{Example}

\allowdisplaybreaks

\title{On a Debiased and Semiparametric Efficient\\ Changes-in-Changes Estimator}

\author[a]{Jinghao Sun\thanks{\texttt{jinghao.sun@pennmedicine.upenn.edu}}}
\author[a,b]{Eric J. Tchetgen Tchetgen\thanks{\texttt{ett@wharton.upenn.edu}}}
\affil[a]{Department of Biostatistics, Epidemiology, \& Informatics, University of Pennsylvania}
\affil[b]{The Statistics and Data Science Department of the Wharton School, University of Pennsylvania}

\begin{document}

\maketitle

\begin{abstract}
  \noindent We present a novel extension of the influential changes-in-changes (CiC) framework \citep{athey2006identification} for estimating the average treatment effect on the treated (ATT) and distributional causal effects in panel data with unmeasured confounding. While CiC relaxes the parallel trends assumption in difference-in-differences (DiD), existing methods typically assume a scalar unobserved confounder and monotonic outcome relationships, and lack inference tools that accommodate continuous covariates flexibly. Motivated by empirical settings with complex confounding and rich covariate information, we make two main contributions. First, we establish nonparametric identification under relaxed assumptions that allow high-dimensional, non-monotonic unmeasured confounding. Second, we derive semiparametrically efficient estimators that are Neyman orthogonal to infinite-dimensional nuisance parameters, enabling valid inference even with machine learning-based estimation of nuisance components. We illustrate the utility of our approach in an empirical analysis of mass shootings and U.S. electoral outcomes, where key confounders, such as political mobilization or local gun culture, are typically unobserved and challenging to quantify.\\

  
\end{abstract}

\section{Introduction}

Do mass shootings influence voter behavior in the United States? Recent years have seen a tragic rise in mass shootings, prompting renewed interest in how such salient events shape public opinion and political outcomes \citep{peterson2024epidemiology}. While some studies suggest that mass shootings lead to decreased Republican vote share in affected areas \citep{yousaf2021sticking}, others argue these findings may reflect methodological artifacts, such as violations of key identification assumptions \citep{hassell2025navigating}. At stake is not only a question of electoral dynamics but also our understanding of how collective trauma and public safety crises influence democratic accountability.

Although the electoral effects of mass shootings offer a timely and consequential case study, they reflect a more general challenge in causal inference: how to credibly estimate treatment effects from panel data when treatment is non-randomly assigned, and unmeasured confounding may be present. These problems arise widely in economics, political science, public health, and policy evaluation \citep{streeter2017adjusting, tchetgen2020introduction, chiu2023causal,arkhangelsky2024causal}, where researchers seek to understand the effects of interventions, shocks, or exposures in settings where randomized experiments are infeasible. Conventional methods like difference-in-differences (DiD) require strong assumptions—such as parallel trends—that may not hold in practice, particularly when treatment assignment or outcome evolution is driven by latent, possibly complex and high-dimensional factors. 

To address these challenges, this paper develops a general framework for robust causal inference in panel data settings with unmeasured confounding and flexible covariate structures. We build on and extend the influential \emph{Changes-in-Changes} (CiC) framework of \citet{athey2006identification}. However, existing CiC methods typically assume a single scalar unmeasured confounder and a monotonic relationship with the outcome. These assumptions may be overly restrictive in settings characterized by complex, nonlinear confounding structures.

Our first contribution is to establish new nonparametric identification results under a relaxed set of assumptions that accommodate high-dimensional, non-monotonic unmeasured confounders. These results enable the identification of the average treatment effect on the treated (ATT), as well as other distributional causal estimands such as quantile treatment effects, without relying on linear trend or monotonicity assumptions.

Our second contribution, building on this foundation, is to develop a semiparametric estimation strategy that achieves both efficiency and robustness in high-dimensional settings. We derive the efficient influence functions (EIFs) for the ATT and for general causal estimands defined through moment conditions, and use them to construct estimators that are Neyman orthogonal to infinite-dimensional nuisance components. This orthogonality ensures that errors in estimating nuisance functions, even when using flexible machine learning methods, do not have a first-order impact on the estimation of the target parameter. To further reduce overfitting bias, we incorporate sample splitting and cross-fitting.

We illustrate the practical utility of our approach through illustrative simulation studies and an empirical application investigating the electoral consequences of mass shootings. The proposed estimator consistently demonstrates strong performance in terms of coverage, bias, and efficiency. While this application serves as a motivating example, the proposed framework is broadly applicable to a wide range of empirical problems in observational panel data settings.

\subsection{Related Literature}

Our work contributes to a broad and rapidly evolving literature on causal estimation in panel data, particularly within the difference-in-differences (DiD) framework. The classical DiD approach relies on the ``parallel trends'' assumption, which states that in the absence of treatment, treated and control units would have followed parallel outcome trajectories \citep{angrist2009mostly, baker2025difference}. Extensions incorporating conditioning on observed covariates relax this assumption somewhat, allowing for heterogeneous trends \citep{heckman1997matching, abadie2005semiparametric, sant2020doubly, callaway2021difference}. Still, the plausibility of parallel trends is often questionable, especially when outcomes are bounded or display nonlinear dynamics \citep{roth2023parallel}.

To address these issues, researchers have proposed various alternative strategies:

\begin{itemize}
  \item \emph{Conditional independence approaches} assume treatment assignment is (conditionally) independent of the untreated potential outcome or its change \citep{fan2012partial, ding2019bracketing, hernan2020causal}.
  
  \item \emph{Bounding strategies} accommodate violations of parallel trends by allowing bounded deviations and using partial identification techniques \citep{rambachan2023more, ye2024negative}.
  
  \item \emph{Instrumental variable methods} introduce auxiliary variables that help correct for potential parallel-trends violations via exogenous variation \citep{richardson2023generalized, ye2023instrumented, zhao2025semiparametric}.
\end{itemize}

A complementary line of research is the \emph{Nonlinear DiD approach}, which impose structure on transformed outcomes or recast identification assumptions on alternative scales \citep{bonhomme2011recovering, puhani2012treatment, callaway2019quantile, park2022universal, wooldridge2023simple, roth2023parallel, fernandez2024distribution}. The present work contributes to this growing literature by extending the Changes-in-Changes (CiC) framework.

The CiC method introduced by \citet{athey2006identification} was a landmark in allowing nonparametric identification of ATT and distributional causal effects without relying on parallel trends. Its identification strategy is built on a ``production function'' model and is invariant to monotonic transformations of the outcome. However, standard CiC implementations typically assume a scalar unmeasured confounder and a monotonic relationship between the confounder and the outcome. Our work builds on this foundation by relaxing these assumptions. We show that identification is still possible under a more general confounding structure, thereby broadening the scope of CiC-style methods in practice.

Additionally, covariates are often crucial for reducing bias from observed confounding. \citet{athey2006identification} outlined a restrictive semiparametric extension that incorporated covariates linearly, without allowing treatment effects to vary across covariate levels. Later approaches \citep{melly2015changes, thome2025estimating} introduced more flexible semiparametric adjustments but continued to rely on linear functional form assumptions, leaving room for model misspecification. By contrast, our framework is fully nonparametric with respect to the observed data, enabling flexible covariate adjustment and allowing treatment effects to vary flexibly, while attaining the semiparametric efficiency bound.

\section{Setting and Notation}

Let $(\Omega, \mathcal{F}, \mathbb{P})$ denote the underlying complete probability space. We observe $n$ independent and identically distributed observations $W_i = (Y_{0,i}, Y_{1,i}, A_i, L_i)$, for $i = 1, \dots, n$, where $Y_{t,i} \in \mathbb{R}$ denotes the outcome at time $t \in \{0,1\}$, $A_i \in \{0,1\}$ is the binary treatment indicator (1 for treated, 0 for control), and $L_i \in \mathbb{R}^p$ is a vector of pre-treatment covariates.

To accommodate unmeasured confounding, we allow for an unobserved random element $U_i$, which may be of arbitrary type (e.g., vector- or function-valued). Under the potential outcomes framework, let $Y_t^{a = \tilde{a}}$ denote the outcome at time $t$ had the unit been assigned treatment $A = \tilde{a}$, possibly contrary to fact. 

The parameter of interest is the \emph{average treatment effect on the treated} (ATT) at time \( t = 1 \), defined as
\[
  \theta = \mathbb{E}\left[Y_1^{a = 1} - Y_1^{a = 0} \,\middle|\, A = 1\right],
\]
where \( \theta \in \Theta \subset \mathbb{R} \). The ATT is particularly relevant in policy evaluation, as it captures the causal effect for the subpopulation that actually received the intervention, without requiring additional assumptions needed for identifying the average treatment effect (ATE). 

Let \(P\) be the true distribution of the observed data. We denote the collection of relevant nuisance parameters by \(\eta \in \mathcal{H} \subseteq \tilde{\mathcal{H}}\), where \(\tilde{\mathcal{H}}\) is a normed space equipped with the \(L^2(P)\) norm, and $\mathcal{H}$ is a convex subset.

For a real-valued random variable \( X \in \mathbb{R} \), we denote its cumulative distribution function (CDF) and quantile function by
\[
  F_X(x) \equiv \Pr(X \le x), \quad
  Q_X(u) \equiv \inf\{x \in \mathbb{R} : F_X(x) \ge u\}, \quad u \in (0,1).
\]
When \( X \) admits a density with respect to a dominating measure (e.g., the Lebesgue measure for continuous outcomes), we denote it by \( f_X(x) \).
Conditional versions given a random element $V$ are $F_{X \mid V}$, $f_{X \mid V}$, and $Q_{X \mid V}$. The support of $X$ under $P$ is denoted by $\supp(X)$.

\section{Identification}

We begin by outlining the proposed identification strategy for the ATT $\theta$. We then extend the results to distributional causal estimands including the quantile treatment effect on the treated and the counterfactual distribution. We then illustrate two common data-generating scenarios that satisfy the identification assumptions. All technical proofs are deferred to the Supplementary Material \ref{appendix:proof}. 

Assumption \ref{assum:setting} defines our setting and are standard in the causal inference \citep{hernan2020causal} and difference-in-differences literature \citep{baker2025difference}.

\begin{assumption}[Setting] \label{assum:setting}
We make the following assumptions:
\begin{enumerate}[label=(\alph*), ref=\theassumption(\alph*)]
    \item \label{assum:cc} {(Causal Consistency)} For $t \in \{0,1\}$, the observed outcome satisfies $Y_{t} = Y_{t}^{a = 1}A + Y_{t}^{a=0}(1-A)$.
    
    \item \label{assum:confounding} {(Latent Unconfoundedness)} For $t = 0, 1$, $A \indep Y_{t}^{a=0} \mid L, U$.
    
    \item \label{assum:pos_U} {(Positivity)} There exists $\delta > 0$ such that $\Prob(A = 1 \mid L, U) \in (\delta, 1 - \delta)$ almost surely.
    
    \item \label{assum:na} {(No Anticipation)} The potential outcomes at time $0$ do not depend on treatment: $Y_{0}^{a=1} = Y_{0}^{a=0}$ almost surely.
\end{enumerate}
\end{assumption}

Assumption~\ref{assum:cc} links observed outcomes to their corresponding potential outcomes.  
Assumption~\ref{assum:confounding} implies that, conditional on both observed covariates \( L \) and latent variables \( U \), treatment assignment is independent of the untreated potential outcome. This contrasts with the general case where \( A \) is not conditionally independent of \( Y_t^{a=0} \) given \( L \) alone.
Assumption~\ref{assum:pos_U} ensures sufficient overlap in treatment assignment across values of $L$ and $U$, enabling regular identification and estimation.  
In some settings, Assumption~\ref{assum:na} may be redundant if one is willing to posit an underlying graphical causal model, such as a single world intervention graph (SWIG) \citep{richardson2013single} or a nonparametric structural equation model with independent errors \citep{pearl2009causality}. These frameworks often imply Assumption~\ref{assum:na} through the temporal ordering of variables, particularly when treatment is assigned strictly between the pre- and post-treatment periods \citep{piccininni2025refining}.

Assumption \ref{assum:ct} pertains to the distributional form of the outcomes and reflects the continuous nature of the outcome considered in this paper.

\begin{assumption}[Continuity] \label{assum:ct}
We assume that, almost surely:
\begin{enumerate}[label=(\alph*), ref=\theassumption(\alph*)]
    \item \label{assum:ct0} The conditional distribution function $F_{Y_{0}\mid A=0, L, U}(\cdot)$ is continuous.
    \item \label{assum:ct1} The conditional distribution function $F_{Y_{1}\mid A=0, L, U}(\cdot)$ is continuous.
\end{enumerate}
\end{assumption}

Assumption~\ref{assum:db} plays a central role in our identification strategy. It allows for causal inference despite the presence of unmeasured confounding, and serves as a relaxation of the traditional parallel trends assumption used in DiD. Let ``$\overset{d}{=}$'' denote equality in distribution.

\begin{assumption}[Distributional Bridge] \label{assum:db}
There exists a function $\gamma(y, l)$, nondecreasing in $y$ for all $l$, such that, almost surely,
\[
\gamma(Y_0^{a = 0}, L)\mid L, U \overset{d}{=} Y_1^{a=0} \mid L, U.
\]
\end{assumption}

Assumption~\ref{assum:db} states that, conditional on covariates and latent variables, the \emph{distribution} of the untreated potential outcome at time $t = 1$ can be characterized by a monotonic transformation of the untreated baseline outcome and observed covariates. This distributional assumption is notably weaker than a rank-preservation condition, such as requiring a strictly monotonic relationship between potential outcomes $Y_0^{a=0}$ and $Y_1^{a=0}$ for each unit. 

The distributional bridge condition is a stronger analog of the \emph{outcome confounding bridge function} used in proximal causal inference \citep{miao2024confounding, cui2024semiparametric}. In that framework, \(Y_0^{a=0}\) is treated as a negative control outcome or an outcome confounding proxy, and one assumes the existence of a square-integrable function \(g(y,l)\) such that
\[
\E[g(Y_0^{a=0}, L) \mid L, U] = \E[Y_1^{a=0} \mid L, U].
\]
A distributional bridge function satisfies this with \(g = \gamma\), but further imposes that for any measurable function \(h(y,l)\),
\[
\E[h(\gamma(Y_0^{a=0}, L), L) \mid L, U] = \E[h(Y_1^{a=0}, L) \mid L, U],
\]
a condition that recovers the full conditional distribution of \(Y_1^{a=0}\).

This added stringency is not merely stronger---it is essential for identifying distributional causal estimands such as quantile treatment effects at arbitrary levels, which require recovering aspects of the outcome distribution beyond its mean. Additionally, it enables identification without requiring a separate \emph{treatment-confounding proxy}, which is typically needed in the proximal causal inference framework.

The next result offers an equivalent and more concrete statement of Assumption~\ref{assum:db}.

\begin{lemma}  \label{lem:QQ}
Under Assumptions~\ref{assum:setting} and~\ref{assum:ct}, Assumption~\ref{assum:db} holds if and only if the quantile-quantile transform
\[
Q_{Y_{1} \mid A=0, L = l, U = u} \circ F_{Y_{0} \mid A=0, L = l, U = u}(\cdot)
\]
is invariant with respect to \(u\) for all values of \(l\). In this case, the function \(\gamma(y, l)\) in Assumption~\ref{assum:db} is given explicitly by
\[
\gamma(y, l) = Q_{Y_{1} \mid A=0, L = l, U = u} \circ F_{Y_{0} \mid A=0, L = l, U = u}(y),
\]
for any \(u\).
\end{lemma}

This characterization shows that Assumption~\ref{assum:db} implies a form of \emph{invariance of conditional optimal transport map} with respect to unmeasured confounders $U$. We further elaborate on this point in Example~\ref{eg:stm}, which introduces a general class of semiparametric transformation models that comply with Assumption~\ref{assum:db}.

Note that Assumptions \ref{assum:setting}-\ref{assum:db} are invariant to monotonic transformations of the outcomes.

We now present our main identification result.

\begin{theorem}[Identification of ATT] \label{thm:idett}
  Under Assumptions~\ref{assum:setting}--\ref{assum:db}, the average treatment effect on the treated (ATT), defined by
  \[
  \theta \equiv \E\left\{ Y_1^{a=1} - Y_1^{a=0} \mid A=1 \right\},
  \]
  is identified by
  \begin{equation} \label{eq:idett}
    \theta = \E\left\{ Y_1 - \gamma(Y_0, L) \mid A=1 \right\} = \E\left\{ Y_1 - Q_{Y_1 \mid A=0, L} \circ F_{Y_0 \mid A=0, L}(Y_0) \mid A=1 \right\}.
  \end{equation}
\end{theorem}

Following the main identification theorem, a couple of remarks are warranted. 

\begin{remark}
  Despite differing in identification assumptions, the ATT identification formula derived by \citet{athey2006identification} coincides with \eqref{eq:idett} in the special case where no covariates $L$ are present.
\end{remark}

\begin{remark}\label{rmk:pt}
Both Assumption~\ref{assum:db} and Assumption 3.1 (Model) of \citet{athey2006identification} (appropriately updated to condition on measured covariates $L$) imply the following condition, which does not explicitly reference the unmeasured confounder \(U\):
\begin{equation}\label{eq:nlpt}
    Q_{Y_1^{a=0} \mid A=1, L} \circ F_{Y_0^{a=0} \mid A=1, L}(\cdot)
    = Q_{Y_1^{a=0} \mid A=0, L} \circ F_{Y_0^{a=0} \mid A=0, L}(\cdot),
\end{equation}
which can be interpreted as a general, nonlinear form of the parallel trends assumption. 

An equivalent formulation is
\begin{equation}\label{eq:nlpt-alt}
    F_{Y_0^{a=0} \mid A=1, L} \circ Q_{Y_0^{a=0} \mid A=0, L}(\cdot)
    = F_{Y_1^{a=0} \mid A=1, L} \circ Q_{Y_1^{a=0} \mid A=0, L}(\cdot),
\end{equation}
which may be viewed as a form of Distributional Equi-Confounding (DEC). This concept, closely related to the conditions introduced by \citet{sofer2016negative}, asserts that the magnitude of hidden confounding---measured in terms of distributional distortions---is the same at \(t = 0\) and \(t = 1\). 
Intuitively, if there is no unmeasured confounding, we have \(Y_t^{a=0} \indep A \mid L\), and the composite mapping 
\(
F_{Y_t^{a=0} \mid A=1, L} \circ Q_{Y_t^{a=0} \mid A=0, L}(y)
\) 
reduces to the identity function \(y\). When confounding is present, this mapping deviates from identity, reflecting distributional distortion induced by unmeasured factors.

Either representation of this condition is sufficient for the identification results in Theorem~\ref{thm:idett} and Corollary~\ref{cor:id_other}, and thus may serve as a primitive identification assumption in place of Assumption~\ref{assum:db} in this paper or Assumption 3.1 of \citet{athey2006identification}, provided the remaining conditions are appropriately adjusted to eliminate explicit reference to the latent confounder \(U\). Nevertheless, explicitly modeling \(U\) is useful for articulating a causal narrative and clarifying the identifying assumptions. Proofs of these claims are provided in the Supplementary Material \ref{appendix:proof}.
\end{remark}

While ATT serves as our primary estimand, the same assumptions allow identification of other causal estimands that may be of independent interest.

\begin{corollary}[Identification of CDT and QTT] \label{cor:id_other}
  Under Assumptions~\ref{assum:setting}--\ref{assum:db}, the following quantities are identified:
  \begin{enumerate}[label=(\alph*), ref=\thecorollary(\alph*)]
    \item \label{cor:idcd} {Counterfactual Distribution on the Treated (CDT)}: For $y \in \supp(Y_1)$, the distribution function
    \[
    \vartheta^{CDT,y} \equiv F_{Y_1^{a=0} \mid A=1}(y)
    \]
    is identified by
    \[
    \vartheta^{CDT,y} = \Prob\left\{ Q_{Y_1 \mid A=0, L} \circ F_{Y_0 \mid A=0, L}(Y_0) < y \mid A=1 \right\}.
    \]

    \item \label{cor:idq} {Quantile Treatment Effect on the Treated (QTT)}: For $\tau \in (0,1)$, the $\tau$-quantile treatment effect
    \[
    \vartheta^{QTT,\tau} \equiv Q_{Y_1^{a=1} \mid A=1}(\tau) - Q_{Y_1^{a=0} \mid A=1}(\tau)
    \]
    is identified by
    \[
    \vartheta^{QTT,\tau} = Q_{Y_1 \mid A=1}(\tau) - \inf\left\{ y : \Prob\left( Q_{Y_1 \mid A=0, L} \circ F_{Y_0 \mid A=0, L}(Y_0) \le y \mid A=1 \right) \ge \tau \right\}.
    \]
  \end{enumerate}
\end{corollary}

To illustrate the scope of our framework, we next describe two classes of data-generating processes (DGPs) that satisfy Assumptions~\ref{assum:setting}, \ref{assum:ct}, and~\ref{assum:db}. These examples concretely demonstrate the conditions under which ATT is identified as per Theorem~\ref{thm:idett}. Formal verification is deferred to Proposition~\ref{prop:semitrans1} in the Supplementary Material \ref{appendix:proof}.

\begin{example}[Semiparametric Transformation Model] \label{eg:stm}
Consider the following semiparametric transformation model. Suppose Assumptions~\ref{assum:cc}, \ref{assum:pos_U}, and \ref{assum:na} hold. Let $U = (U_1, \ldots, U_q) \in \mathbb{R}^q$ denote a vector of unmeasured confounders, which may consist of arbitrary types. The model is specified as follows:

\begin{enumerate}
  \item {Treatment assignment:} Let $h$ be an arbitrary $\{0,1\}$-valued function, and let $\delta$ be a mean-zero random noise. Treatment is assigned via
  \[
    A = h(L, U, \delta).
  \]

  \item {Potential outcome model:} For each $t \in \{0,1\}$,
  \[
    Y_t^{a=0} = \beta_t(k_t(L) + m(U,L) + \epsilon_t),
  \]
  where $\beta_t(\cdot)$ is a strictly increasing but otherwise unrestricted time-specific transformation function; $k_t$ is a time-specific function of the observed covariates $L$; and $m(\cdot,\cdot)$ satisfies $m(u_0,\cdot)=0$ for a fixed reference value $u_0$, but otherwise is unrestricted in both its arguments. 

  \item {Error structure:} The error term $\epsilon_t$ satisfies the following conditions:
  \[
    \epsilon_t \indep U \mid A=0,L, \quad \delta \indep \epsilon_t \mid L,U, \quad \epsilon_0 \mid A=0, L, U \overset{d}{=} \epsilon_1 \mid A=0, L, U.
  \]
  Let $F_{\epsilon \mid A=0,L, U}$ denote the common conditional distribution of $\epsilon_0$ and $\epsilon_1$ given $A=0, L,$ and $U$. This distribution is assumed to be continuous, strictly increasing, and have mean zero.
\end{enumerate}

Under this model, Theorem~\ref{thm:idett} implies
\[
\theta = \E\left\{ Y_1 - Q_{Y_1 \mid A=0, L} \circ F_{Y_0 \mid A=0, L}(Y_0) \mid A=1 \right\} = \E\left[ Y_1 - \beta_1\left\{ \beta_0^{-1}(Y_0) + k_1(L) - k_0(L) \right\} \mid A=1 \right].
\]

We highlight that no assumptions are placed on the treated potential outcome $Y_1^{a=1}$. 
As a concrete instantiation of our broader framework, this semiparametric model illustrates both the complexity of confounding that can be accommodated and the functional form that the quantile-quantile mapping $\gamma$ may take.

\end{example}

\begin{example}[Linear Model / Difference-in-Differences] \label{eg:did}
Consider a special case of Example~\ref{eg:stm} in which the transformation functions are the linear (identity) mappings: $\beta_0(x) = \beta_1(x) = x$, with measured covariates $L$ omitted for simplicity. In this case, $k_t$ are time-specific intercepts, and 
\begin{align*}
    Q_{Y_1 \mid A=0} \circ F_{Y_0 \mid A=0}(Y_0) &= Y_0 + k_1 - k_0 = Y_0 + \E(Y_1 - Y_0 \mid A=0).
\end{align*}
Then, by Theorem~\ref{thm:idett}, the ATT is given by
\begin{align*}
    \theta &= \E(Y_1^{a=1} - Y_1^{a=0} \mid A=1) \\
    &= \E\left\{ Y_1 - Q_{Y_1 \mid A=0} \circ F_{Y_0 \mid A=0}(Y_0) \mid A=1 \right\} \\
    &= \E(Y_1 - Y_0 \mid A=1) - \E(Y_1 - Y_0 \mid A=0),
\end{align*}
which recovers the standard \emph{difference-in-differences} (DiD) formula.
\end{example}

\section{Estimation}

The identification results established in the previous section provide a foundation for causal inference under our proposed framework. However, identification alone is not sufficient for developing practical estimators with desirable statistical properties, especially when there are (possibly high-dimensional) continuous covariates. To this end, we turn to modern semiparametric efficiency theory. We develop an efficient influence function-based estimator for the causal estimands and establish its theoretical properties.

\subsection{Efficient Influence Function}

A central tool in this theory is the efficient influence function (EIF), which serves multiple purposes. First, it provides straightforward characterization of the semiparametric efficiency bound—the lowest possible asymptotic variance—for regular and asymptotically linear (RAL) estimators of the target parameter under the specified model. Second, the EIF allows for the construction of a RAL estimator that attains the efficiency bound via solving estimating equations. Third, and crucially in modern applications, efficient-influence-function-based estimators are Neyman orthogonal \citep{neyman1959optimal} to nuisance parameters, making them robust to estimation errors in nuisance components. This robustness is critical when using flexible machine learning methods to estimate high-dimensional or nonparametric nuisance functions, as it mitigates the impact of slow convergence rates. These advantages have been widely discussed in foundational and recent literature (e.g., \citealp{newey1990semiparametric}; \citealp{bickel1993efficient}; \citealp{robins2008higher}; \citealp{chernozhukov2018double}).

In what follows, we derive the EIF for ATT based on the identification results in Theorem \ref{thm:idett}. We then generalize the approach to a class of causal parameters defined via moment restrictions, including counterfactual outcome distributions and quantile treatment effects on the treated (CDT and QTT) as examples.

We use the notation \(\EIF(W; \theta, \eta)\) to denote the influence function associated with the target parameter \(\theta\), evaluated at a data point \(W\), where \(\eta\) represents a collection of nuisance functions. 
Define \(\gamma(y, l) \equiv Q_{Y_1 \mid A = 0, L=l} \circ F_{Y_0 \mid A=0, L=l}(y)\) and \(\nu(x,l) \equiv \frac{\Prob(A = 1 \mid \gamma(Y_0, l) = x, L=l)}{\Prob(A = 0 \mid \gamma(Y_0, l) = x, L=l)}\), and let \(\pi \equiv \Prob(A = 1)\). Then we have the following result.

\begin{theorem}[Efficient Influence Function of ATT] \label{thm:eif}
  Under Assumptions \ref{assum:setting}--\ref{assum:db}, the Efficient Influence Function (EIF) for \(\theta\), as identified in Theorem \ref{thm:idett} under a nonparametric model over the observed data, is
  \[
    \EIF(W;\theta, \eta) = \frac{A}{\pi}\left[Y_1 - \gamma(Y_0, L) - \theta\right] + \frac{1 - A}{\pi}\int_{Y_1}^{\gamma(Y_0, L)} \nu(x, L) \, \dx{x},
  \]
  where $\eta = (\gamma, \nu, \pi)$. 
  Consequently, the corresponding semiparametric efficiency bound for \(\theta\) is \(\mathbb{E}\left\{\EIF(W;\theta, \eta)^2\right\}\).
\end{theorem}

We now extend the previous result to a broad class of estimands on the treated, denoted by \(\vartheta\), which are identified through a moment condition of the form:
\begin{equation} \label{eq:estimand}
  \mathbb{E}\left[g(W; \vartheta, \gamma) \mid A = 1\right] = 0,
\end{equation}
where the function \(g(W; \vartheta, \gamma)\) takes the form
\[
  g(W; \vartheta, \gamma) = \tilde{g}(\gamma(Y_0, L), \vartheta),
\]
with \(\tilde{g}(\cdot, \cdot)\) known, right-continuous, and of bounded variation in its first argument. Several important causal estimands fall into this framework. For example:

\begin{enumerate}
  \item {Counterfactual mean on the treated:}
  For \(\vartheta = \mathbb{E}[Y_{1}^{a=0} \mid A = 1]\),
  \[
    g(W; \vartheta, \gamma) = \gamma(Y_0, L) - \vartheta.
  \]

  \item {Counterfactual distribution function:}
  For \(\vartheta = F_{Y^{a=0}_1 \mid A=1}(y)\),
  \[
    g(W; \vartheta, \gamma) = \Ind\{\gamma(Y_0, L) < y\} - \vartheta.
  \]

  \item {Counterfactual quantile on the treated:} For \(\vartheta = Q_{Y^{a=0}_1 \mid A=1}(\tau)\), assuming sufficient regularity of the quantile \citep{firpo2007efficient},
  \[
    g(W; \vartheta, \gamma) = \Ind\{\gamma(Y_0, L) < \vartheta\} - \tau.
  \]
\end{enumerate}

The following theorem characterizes the efficient influence function (EIF) for this general class of estimands.

\begin{theorem}[Efficient Influence Function for General Estimands] \label{thm:eif_gen}
Under Assumptions~\ref{assum:setting}--\ref{assum:db}, the efficient influence function (EIF) for the estimand \(\vartheta\), identified by Equation~\eqref{eq:estimand}, is given by
\[
  \EIF(W;\vartheta, \eta) = \frac{
    A\,g(W;\vartheta, \gamma)
    + (A - 1) \int_{\left(Y_1, \gamma(Y_0, L)\right]} \nu(x, L) \, \mathrm{d}_x \tilde{g}(x, \vartheta)
  }{
    -\pi\,\left.\partial_{\vartheta'} \mathbb{E}[g(W;\vartheta', \gamma) \mid A = 1]\right|_{\vartheta' = \vartheta}
  },
\]
where the integral is a Lebesgue--Stieltjes integral of \(\nu\) with respect to \(\tilde{g}\) in \(x\), and \(\eta\) denotes the collection of nuisance parameters evaluated at their true values. The semiparametric efficiency bound for \(\vartheta\) is \(\mathbb{E}\left\{ \EIF(W;\vartheta, \eta)^2 \right\}\).
\end{theorem}

We now specialize this result to the estimands in Corollary~\ref{cor:id_other}. Define the sign function \(\operatorname{sign}(x) = \Ind\{x > 0\} - \Ind\{x < 0\}\), and introduce the compound sign function
\[
  \chi(x, W; \gamma) \equiv \operatorname{sign}(Y_1 - \gamma(Y_0, L)) \cdot \Ind\left\{ \min(Y_1, \gamma(Y_0, L)) \le x \le \max(Y_1, \gamma(Y_0, L)) \right\},
\]
which takes values in \(\{-1, 0, 1\}\) and simplifies the representation of the EIF.

\begin{corollary}[Efficient Influence Functions for CDT and QTT] \label{cor:eif_other}
Under Assumptions~\ref{assum:setting}--\ref{assum:db} and a nonparametric model over the observed data:
\begin{enumerate}[label=(\alph*), ref=\thecorollary(\alph*)]
  \item \label{cor:eif_cd} \textbf{CDT:} For \(\vartheta^{CDT,y} = F_{Y^{a=0}_1 \mid A=1}(y)\), the EIF is
  \[
    \EIF(W;\vartheta^{CDT,y}, \eta) = \frac{A}{\pi} \left[ \Ind\{\gamma(Y_0, L) < y\} - \vartheta^{CDT,y} \right]
    + \frac{A - 1}{\pi} \nu(y, L) \chi(y, W; \gamma),
  \]
  where \(\eta = (\gamma, \nu, \pi)\). The corresponding semiparametric efficiency bound is \(\mathbb{E}\left\{ \EIF(W;\vartheta^{CDT,y}, \eta)^2 \right\}\).

  \item \label{cor:eif_q} \textbf{QTT:} For quantile level \(\tau \in (0,1)\), assuming regularity of the quantile \citep{firpo2007efficient}, the EIF for \(\vartheta^{QTT,\tau} = Q_{Y_1^{a=1} \mid A=1}(\tau) - Q_{Y_1^{a=0} \mid A=1}(\tau)\) is
  \begin{align*}
    \EIF(W;\vartheta^{QTT,\tau}, \eta) &= \frac{A}{\pi} \frac{\Ind\{Y_1 \le \vartheta_1\} - \tau}{-f_{Y_1 \mid A=1}(\vartheta_1)} - \frac{
      A\left[ \Ind\{\gamma(Y_0, L) < \vartheta_2\} - \tau \right]
      + (A - 1) \nu(\vartheta_2, L) \chi(\vartheta_2, W; \gamma)
    }{
      -\pi f_{\gamma(Y_0, L) \mid A=1}(\vartheta_2)
    },
  \end{align*}
  where \(\vartheta_1 = Q_{Y_1^{a=1} \mid A=1}(\tau)\), \(\vartheta_2 = Q_{Y_1^{a=0} \mid A=1}(\tau)\), and \(\eta = (\gamma, \nu, \pi, f_{Y_1 \mid A=1}, f_{\gamma(Y_0, L) \mid A=1})\). The corresponding semiparametric efficiency bound is \(\mathbb{E}\left\{ \EIF(W;\vartheta^{QTT,\tau}, \eta)^2 \right\}\).
\end{enumerate}
\end{corollary}

\subsection{Debiased Semiparametric Estimator}

We construct an estimator for the average treatment effect on the treated (ATT), denoted by \(\theta\). Since the efficient influence function (EIF) \(\EIF(W; \theta, \eta)\) satisfies the moment condition \(\mathbb{E}[\EIF(W; \theta, \eta)] = 0\), it can be used as a valid estimating function. Accordingly, we define \(\psi(W; \tilde{\theta}, \tilde{\eta})\) to share the same functional form as \(\EIF(W; \theta, \eta)\), where \(\tilde{\theta} \in \Theta\) and \(\tilde{\eta} \in \mathcal{H}\) represent candidate values for the target parameter and nuisance components, respectively. Estimators for more general causal estimands defined by Equation~\eqref{eq:estimand} can be constructed analogously using their corresponding EIFs from Theorem~\ref{thm:eif_gen}. For brevity, we focus here on ATT.

Our estimation strategy employs sample splitting and cross-fitting \citep{schick1986asymptotically, chernozhukov2018double} to mitigate overfitting bias. Moreover, the use of the Neyman-orthogonal EIF from Theorem~\ref{thm:eif} helps to control regularization bias arising from the estimation of complex or high-dimensional nuisance functions. The resulting estimator is semiparametric efficient and supports valid inference even when nuisance quantities are learned via flexible machine learning methods with convergence rates slower than \(\sqrt{n}\), under mild regularity conditions.

To formalize the construction, let \(\{W_i\}_{i=1}^n\) denote an i.i.d. sample. Define the index set \([n] \equiv \{1, \ldots, n\}\), and fix an integer \(K\) denoting the number of folds used for cross-fitting. Assume for simplicity that \(n\) is divisible by \(K\), and let \((\mathcal{I}_k)_{k \in [K]}\) be a random, equal-sized partition of \([n]\). For each \(k\), let \(\mathcal{I}_{(-k)} \equiv \bigcup_{k' \neq k} \mathcal{I}_{k'}\) denote the training set excluding fold \(k\).

Nuisance functions are estimated using only data from \(\mathcal{I}_{(-k)}\), yielding estimates \[\hat{\eta}_k = \hat{\eta}((W_i)_{i \in \mathcal{I}_{(-k)}}; \zeta),\] where \(\hat{\eta}\) may be obtained via machine learning or nonparametric methods, and \(\zeta\) is a tuning parameter chosen either a priori or via cross-validation. Let \(z_{\alpha}\) denote the \((1 - \alpha)\)-quantile of the standard normal distribution \(\mathcal{N}(0,1)\). The full estimation and inference procedure is summarized in Algorithm~\ref{alg:cf}.

\begin{algorithm}[ht]
  \caption{Cross-Fitted Estimator \(\hat{\theta}\) and Confidence Interval \(\CI_{n,\alpha}\)}
  \label{alg:cf}
  \KwIn{Significance level \(\alpha \in (0,1)\); data \(\{W_i\}_{i=1}^n\); number of folds \(K\); optional model selection folds \(K'\); number of repetitions \(S\).}
  \KwOut{Point estimate \(\hat{\theta}\), variance estimate \(\hat{\sigma}^2\), confidence interval \(\CI_{n,\alpha}\).}
  
  \For{\(s = 1\) \KwTo \(S\)}{
    \textbf{1.1} Randomly partition \([n]\) into \(K\) folds: \((\mathcal{I}_k^s)_{k \in [K]}\). \\
    \textbf{1.2} For each \(k \in [K]\): \\
    \Indp
    (a) Select tuning parameter \(\hat{\zeta}_k^s\) via \(K'\)-fold cross-validation on \(\mathcal{I}_{(-k)}^s\), or set to a pre-specified \(\zeta\). \\
    (b) Estimate nuisance functions: \(\hat{\eta}_k^s = \hat{\eta}((W_i)_{i \in \mathcal{I}_{(-k)}^s}; \hat{\zeta}_k^s)\). \\
    \Indm
    \textbf{1.3} Solve for \(\tilde{\theta}^s\) using:
    \[
      \sum_{k=1}^K \sum_{i \in \mathcal{I}_k^s} \psi(W_i; \tilde{\theta}^s, \hat{\eta}_k^s) = 0.
    \]\\
    \textbf{1.4} Compute variance estimate:
    \[
      \tilde{\sigma}^{2,s} = \frac{1}{n} \sum_{k=1}^K \sum_{i \in \mathcal{I}_k^s} \psi^2(W_i; \tilde{\theta}^s, \hat{\eta}_k^s).
    \]
  }

  \textbf{2. Median adjustment:}
  \[
    \hat{\theta} = \med\left(\{\tilde{\theta}^s\}_{s=1}^S\right), \quad
    \hat{\sigma}^2 = \med\left(\{\tilde{\sigma}^{2,s} + (\tilde{\theta}^s - \hat{\theta})^2\}_{s=1}^S\right).
  \]

  \textbf{3. Compute \((1-\alpha)\)-confidence interval:}
  \[
    \CI_{n,\alpha} = \left[\hat{\theta} - z_{\alpha/2} \frac{\hat{\sigma}}{\sqrt{n}}, \; \hat{\theta} + z_{\alpha/2} \frac{\hat{\sigma}}{\sqrt{n}}\right].
  \]
\end{algorithm}

\begin{remark}[Nuisance Estimators]
Let \(\hat{Q}_{Y_1 \mid A = 0, L}\) and \(\hat{F}_{Y_0 \mid A = 0, L}\) denote estimators for the conditional quantile and distribution functions, respectively. The estimated transformation is given by \(\hat{\gamma}(y, L) = \hat{Q}_{Y_1 \mid A = 0, L} \circ \hat{F}_{Y_0 \mid A = 0, L}(y)\). To estimate \(\nu(x, L)\), one may regress \(A\) on \(\hat{\gamma}(Y_0, L)\) and \(L\), and compute the implied odds. A wide range of machine learning and nonparametric methods are available for estimating conditional quantile and distribution functions, with well-established theoretical guarantees; see, for example, \citet{he1994convergence, hall1999methods, meinshausen2006quantile, ishwaran2008random, li2008nonparametric}.
\end{remark}

\subsection{Asymptotic Properties}

We now analyze the asymptotic properties of the proposed estimator for the average treatment effect on the treated (ATT), denoted by \(\theta\). These results justify using the efficient influence function (EIF) for point estimation and Wald-type inference. While we focus on ATT for brevity of exposition, analogous asymptotic guarantees extend to general estimands defined by Equation~\eqref{eq:estimand}, including quantile treatment effects, as a consequence of the Neyman orthogonality of the EIFs established in Theorem~\ref{thm:eif_gen}; see \citet{chernozhukov2018double}. Alternatively, confidence intervals can be obtained by inverting test statistics based on the EIF, as recently illustrated by \citet{lee2025inference} in the instrumental variable setting.

Recall that \(\eta = (\gamma, \nu, \pi)\) denotes the true nuisance functions, and let \(\tilde{\eta} = (\tilde{\gamma}, \tilde{\nu}, \tilde{\pi})\) represent a generic element in \(\mathcal{H} \subseteq \tilde{\mathcal{H}}\). We endow \(\tilde{\mathcal{H}}\) with the \(L^2(P)\) norm defined by:
\[
\|\tilde{\eta}\| \equiv \|\tilde{\eta}\|_{L^2(P)} = \left( \mathbb{E}\left[ \tilde{\gamma}(Y_0, L)^2 + \left\{ \tilde{\nu}(\gamma(Y_0, L), L) \right\}^2 + \tilde{\pi}^2 \right] \right)^{1/2}.
\]
We also define the norms for each component relative to their respective function spaces:
\[
\|\tilde{\gamma}\| = \left( \mathbb{E}\left[ \tilde{\gamma}(Y_0, L)^2 \right] \right)^{1/2}, \quad
\|\tilde{\nu}\| = \left( \mathbb{E}\left[ \left\{ \tilde{\nu}(\gamma(Y_0, L), L) \right\}^2 \right] \right)^{1/2}, \quad
\|\tilde{\pi}\| = |\tilde{\pi}|.
\]
We begin by formalizing conditions on the nuisance function class and the estimation rate of the nuisance estimators.

\begin{assumption}[Nuisance Function Regularity and Estimation Rates] \label{assum:regularity}
We assume:
\begin{enumerate}[label=(\alph*), ref=\theassumption(\alph*)]
  \item \label{assum:nu} For all \(\tilde{\eta} = (\tilde{\gamma}, \tilde{\nu}, \tilde{\pi}) \in \mathcal{H}\), $x \in \supp(Y_1\mid A = 0)$, and $l \in \supp(L)$, the function \(\tilde{\nu}(x, l)\) is continuously differentiable in \(x\), and there exists a constant $C$ such that \(|\partial_x\tilde{\nu}(x, l)| \le C < \infty\). 
  \item \label{assum:rate} The nuisance estimator \(\hat{\eta}_k^s\) is consistent and satisfies
  \[
  \|\hat{\eta}_k^s - \eta\| = o(n^{-1/4}).
  \]
\end{enumerate}
\end{assumption}

To clarify the robustness of the estimator, we first characterize the bias induced by deviation from the true nuisance functions via the first and second Gateaux derivatives of the moment function \(\mathbb{E}[\psi(W; \theta, \eta)]\). Let $A \lesssim B$ denote that $A \le c\cdot B$ for some constant $c > 0$.

\begin{lemma}[Bias Structure] \label{lem:bias}
Let \(\eta_\lambda = \eta + \lambda(\tilde{\eta} - \eta) = (\gamma_\lambda, \nu_\lambda, \pi_\lambda)\) for \(\lambda \in [0,1)\), and define
\[
\Phi(\lambda) \equiv \mathbb{E}\left[\psi(W; \theta, \eta_\lambda)\right].
\]
Let \(\Delta \gamma = \tilde{\gamma} - \gamma\), \(\Delta \nu = \tilde{\nu} - \nu\), and \(\Delta \pi = \tilde{\pi} - \pi\). Then, we have
\begin{enumerate}[label=(\alph*), ref=\theassumption(\alph*)]
  \item {(Neyman Orthogonality)} \quad \(\Phi'(0) = 0\).
  \item {(Second-Order Bias)}  For all $\lambda \in [0,1)$, 
  \[
    \Phi''(\lambda)\ \lesssim\ 
    \|\Delta\gamma\|\,\|\Delta\nu\|
    +\|\Delta\gamma\|^2
    +|\Delta\pi|\|\Delta\gamma\|
    +(\Delta\pi)^2 \ \lesssim\  
    \|\Delta\eta\|^2.
  \]  
\end{enumerate}
\end{lemma}
We now present the main asymptotic result. Let \(\mathcal{P}\) denote the collection of data-generating processes that satisfy Assumptions~\ref{assum:setting}--\ref{assum:regularity}.

\begin{theorem}[Asymptotic Properties] \label{thm:asymptotic}
Under Assumptions~\ref{assum:setting}--\ref{assum:regularity}, the cross-fitted estimator \(\hat{\theta}\) constructed via Algorithm~\ref{alg:cf} satisfies, as \(n \to \infty\):
\begin{enumerate}[label=(\alph*), ref=\thetheorem(\alph*)]
  \item {Consistency:} \quad \(\hat{\theta} \pto \theta\).
  \item {Asymptotic Normality and Semiparametric Efficiency:}
  \[
  \sqrt{n}(\hat{\theta} - \theta) = \frac{1}{\sqrt{n}} \sum_{i = 1}^{n} \EIF(W_i; \theta, \eta) + o_p(1) \;\; \dto \; \mathcal{N}(0, \sigma^2),
  \]
  uniformly over \(\mathcal{P}\), where \(\sigma^2 = \mathbb{E}[\EIF(W; \theta, \eta)^2]\) achieves the semiparametric efficiency bound.
  \item {Variance Estimator Consistency:} \quad \(\hat{\sigma}^2 \pto \sigma^2\).
  \item {Uniformly Valid Confidence Interval:} \quad 
  \[
  \sup_{P \in \mathcal{P}} \left| \mathbb{P}_P\left(\theta \in \CI_{n,\alpha} \right) - (1 - \alpha) \right| \to 0.
  \]
\end{enumerate}
\end{theorem}

Hypothesis testing for $\theta$ can proceed via the constructed confidence interval. Bootstrap procedures could offer further refinements; we leave their investigation to future work.

\section{Simulation}
\paragraph{Data-Generating Process}  
We simulate data from the semiparametric transformation model introduced in Example~\ref{eg:stm} to assess the performance of various estimators. Complete simulation specifications, including all functional forms and parameter settings, are provided in Supplementary Material~\ref{appendix:sim}.

Covariates \( L \in \mathbb{R}^p \) with \( p = 6 \) are independently drawn from a standard multivariate normal distribution. Treatment assignment is governed by a nonlinear function of both observed and unobserved variables:
\begin{equation}\label{eq:simu-1}
  A \sim \mathrm{Bernoulli}(\expit(\phi)), \quad \phi = c_1 + \frac{c_2}{\sqrt{p}} \sum_{j=1}^p (L_j^2 - 1) + U_1 + U_2^2,
\end{equation}
where \( U_1, U_2 \sim \mathcal{N}(0,1) \) are independent unmeasured confounders, and \( \expit(x) \equiv 1/(1 + e^{-x}) \) denotes the logistic function.

Potential outcomes under no treatment evolve over two time periods \( t = 0,1 \) according to:
\begin{equation}\label{eq:simu-2}
  Y_t^{a=0} = \beta_t\bigl(k_t(L) + m(U, L) + \varepsilon_t\bigr), \quad \beta_t(y) = \frac{1}{1 + \exp(-y / \alpha_t)},
\end{equation}
where \( \varepsilon_t \sim \mathcal{N}(0,1) \) are independent error terms, and \( U \equiv (U_1, U_2) \). The functions \( k_t(\cdot) \) and \( m(\cdot, \cdot) \) capture nonlinear relationships with covariates and unobservables.

To simplify exposition, we impose the sharp null hypothesis for treated potential outcomes: \( Y_t^{a=1} = Y_t^{a=0} \). Under this assumption, the true average treatment effect on the treated (ATT) is zero in the simulated data-generating process.

\paragraph{Setup}  
We simulate datasets of sizes \( n = 500,\, 1000,\, 2000 \), repeating each scenario 1000 times. The following three estimators are compared:

\begin{enumerate}
  \item \texttt{Debiased CiC}: Our proposed estimator, where nuisance functions are estimated using random forests \citep{breiman2001random, generalized2019athey}. We implement five-fold cross-fitting and use five-fold cross-validation for model selection. Median adjustment is applied with \( S = 20 \) repetitions.
  
  \item \texttt{CiC}: To heuristically illustrate the role of debiasing in the presence of covariates, we evaluate a variant that uses only the identification formula. Specifically, Step 1.3 in Algorithm~\ref{alg:cf} is replaced with the estimating equation
  \[
  \sum_{k=1}^K \sum_{i \in \mathcal{I}_k^s} \check\psi(W_i; \tilde{\theta}^s, \hat{\eta}_k^s) = 0,
  \]
  where \( \check\psi(W; \theta, \eta) = \frac{A}{\pi} \left\{ Y_1 - \gamma(Y_0, L) - \theta \right\} \) corresponds to the identification formula without correction. Step 1.4 is skipped, and the variance estimator \( \hat{\sigma}^2 \) from \texttt{Debiased CiC} is reused, since \texttt{CiC} lacks a valid standard error estimator when covariates are included. As shown in Table \ref{tab:simresults}, the standard deviations of \texttt{CiC} and \texttt{Debiased CiC} are nearly identical, supporting this heuristic. All other steps mirror those of \texttt{Debiased CiC}.
  
  \item \texttt{DiD}: A conventional difference-in-differences estimator assuming conditional parallel trends, implemented using the \texttt{did} package \citep{callaway2021did} with its default parametric linear specification.
\end{enumerate}

\begin{center}
\begin{minipage}[b]{0.48\textwidth}
  \vspace*{0pt}
    \centering
    \includegraphics[width=0.9\linewidth]{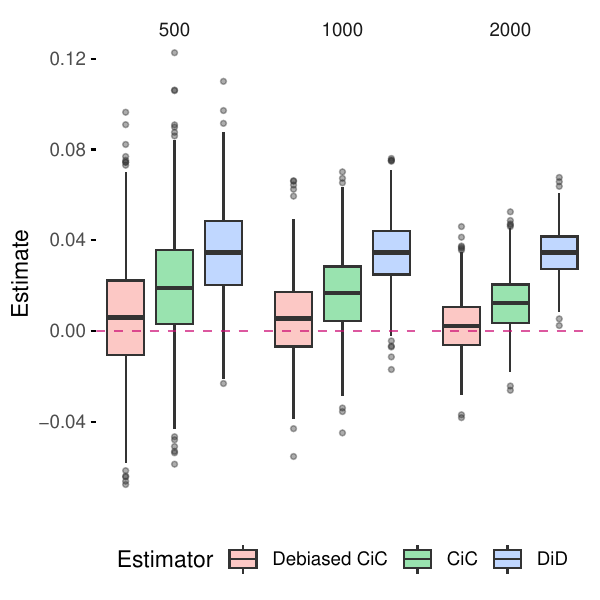}
    \captionof{figure}{Boxplots of ATT estimates from 1{,}000 simulation replications at sample sizes \( n = 500, 1000, 2000 \), comparing the proposed \texttt{Debiased CiC} estimator with \texttt{CiC} and \texttt{DiD}. The dashed red line marks the true ATT, set to zero by design. The \texttt{Debiased CiC} estimator consistently concentrates around the truth with minimal bias, while \texttt{DiD} remains biased and \texttt{CiC} shows slow bias decay.}
    \label{fig:simboxplot}
\end{minipage}
\hfill
\begin{minipage}[b]{0.48\textwidth}
  \vspace*{0pt} 
  \centering
    \footnotesize
    \begin{tabular}{@{}clclc@{}}
        \toprule
        \textbf{$n$} & \textbf{Metric} & \textbf{Debiased CiC} & \textbf{CiC} & \textbf{DiD} \\
        \midrule
        \multirow{4}{*}{500}  & Coverage & 0.940 & 0.865 & 0.601 \\
                              & Bias     & 0.006 & 0.019 & 0.035 \\
                              & SE       & 0.024 & 0.024 & 0.020 \\
                              & SD       & 0.025 & 0.025 & 0.021 \\
        \midrule
        \multirow{4}{*}{1000} & Coverage & 0.937 & 0.809 & 0.341 \\
                              & Bias     & 0.005 & 0.016 & 0.034 \\
                              & SE       & 0.017 & 0.017 & 0.014 \\
                              & SD       & 0.018 & 0.018 & 0.014 \\
        \midrule
        \multirow{4}{*}{2000} & Coverage & 0.935 & 0.808 & 0.083 \\
                              & Bias     & 0.002 & 0.012 & 0.034 \\
                              & SE       & 0.012 & 0.012 & 0.010 \\
                              & SD       & 0.012 & 0.012 & 0.010 \\
        \bottomrule
    \end{tabular}
    \vspace{2.5em}
    \captionof{table}{Quantitative comparison of estimator performance across simulation settings, showing coverage at the 95\% nominal level, bias, average asymptotic standard error (SE), and empirical standard deviation (SD). The \texttt{Debiased CiC} estimator maintains accurate coverage and low bias across all sample sizes, while \texttt{CiC} and \texttt{DiD} suffer from undercoverage due to slower bias reduction and misspecification, respectively.}
    \label{tab:simresults}
\end{minipage}
\end{center}

\paragraph{Results}  
Table~\ref{tab:simresults} and Figure~\ref{fig:simboxplot} summarize the performance of the estimators across key metrics: coverage probability at the 95\% nominal level, bias, average estimated asymptotic standard error (SE), and empirical standard deviation (SD) of the estimates across 1000 replications.

The \texttt{Debiased CiC} estimator consistently achieves coverage close to the nominal level (93.5\% - 94.0\%), exhibits negligible bias, and produces standard error estimates that closely align with empirical standard deviations—even at smaller sample sizes. 

In contrast, both \texttt{CiC} and \texttt{DiD} exhibit deteriorating coverage as \( n \) increases. The \texttt{DiD} estimator suffers from persistent bias stemming from violations of the parallel trends assumption, resulting in coverage probabilities that approach zero at larger sample sizes. While the \texttt{CiC} estimator shows diminishing bias as \( n \) increases, the convergence is too slow to ensure valid inference, leading to substantial undercoverage across all sample sizes considered ($80.8\% - 86.5\%$).

Overall, these results highlight the advantages of the proposed \texttt{Debiased CiC} estimator in delivering reliable inference and low bias under complex, nonlinear, and confounded data-generating processes.

\section{Application: Mass Shootings and Electoral Outcomes}\label{sec:app}

\paragraph{Background} 
Mass shootings are among the most visible and devastating forms of gun violence in the United States, inflicting profound harm on communities and dominating national political discourse \citep{peterson2024epidemiology}. Despite widespread public support for gun control, policy responses remain limited and inconsistent. This disconnect between tragedy and legislative action presents a central puzzle in American politics \citep{hassell2025navigating}.
These events, amplified by media coverage, often galvanize political attention and mobilize voters, making them more than isolated tragedies. They can shift public attitudes, influence political behavior, and alter electoral incentives. Studying their electoral impact thus provides a lens into broader structural forces shaping gun policy. 

Prior studies use panel data and difference-in-differences (DiD) designs to estimate these effects, but results are mixed \citep{hassell2020mobilize, yousaf2021sticking, garcia2022violence, hassell2025navigating}. Many rely on the strong parallel trends assumption and face challenges from confounding. Counties vary across both observable traits (e.g., demographics, economics, geography) and characteristics that are difficult to measure or quantify (e.g., media narratives, community activism, local gun culture). While the former can be controlled for, the latter are often high-dimensional and may influence outcomes in complex, nonlinear ways.

We address these challenges with the proposed approach that offers two key advantages: (1) it relaxes the parallel trends assumption, enabling more credible causal inference; and (2) it leverages machine learning to incorporate a rich set of covariates, thereby mitigating confounding bias. The empirical analysis below demonstrates the practical value of this methodology in addressing urgent, real-world policy questions.

\paragraph{Data Description \& Setup} We re-analyze the dataset from \citet{yousaf2021sticking, yousaf2022replication}, extending the original specification by incorporating a richer set of relevant covariates. The sample consists of approximately 3,000 U.S. counties observed in the 2004 (pre-treatment) and 2008 (post-treatment) presidential election years, a cycle selected for its consistency and completeness across data sources. The treatment indicator equals one for counties that experienced a mass shooting—successful or failed—between 2005 and 2008, where a successful mass shooting follows the FBI definition of an event resulting in four or more deaths at a single location. In total, 65 counties were classified as treated. The outcome is the county-level Republican vote share in the presidential elections. Covariates, measured prior to treatment, include economic indicators (household income, Gini coefficient), demographic characteristics (population size, proportion never married, racial Herfindahl-Hirschman index), education (proportion of college dropouts), health measures (proportion of residents reporting mental health issues), and geographic location (state).

\paragraph{Estimation \& Results}

We estimate the effect of mass shootings using two methods: (a) the proposed Changes-in-Changes (\texttt{Debiased CiC}) estimator and (b) the standard Difference-in-Differences (\texttt{DiD}) estimator, which relies on the (conditional) parallel trends assumption and is implemented in the \textsf{R} package \texttt{did} \citep{callaway2021did}. For each method, we consider two specifications: with and without covariates. Notably, \texttt{Debiased CiC} without covariates is asymptotically equivalent to the CiC estimator described in Section 5 of \citet{athey2006identification}.

Table~\ref{tab:results} reports point estimates, standard errors (SE), and 95\% confidence intervals for the ATT using the four estimators described above. Our covariate-adjusted \texttt{Debiased CiC} estimate — smaller in magnitude and statistically indistinguishable from zero — suggests a muted electoral response. A comparison across estimators highlights three key insights:

First, \texttt{Debiased CiC} consistently yields smaller estimated effects than \texttt{DiD}, both with and without covariates. This attenuation supports the concern that \texttt{DiD} may overstate treatment effects when the parallel trends assumption is violated. In contrast, \texttt{Debiased CiC} relaxes this assumption, allowing for more flexible, nonlinear trends, and potentially producing more credible estimates.

Second, covariate adjustment substantially influences the estimated effects. Without covariates, the ATT is estimated at $-2.60\%$ under \texttt{DiD} and $-1.63\%$ under \texttt{Debiased CiC}. Including covariates reduces the magnitudes of both estimates by over one percentage point, underscoring the importance of addressing confounding.

Third, \texttt{Debiased CiC} with covariates achieves a lower standard error (SE = 0.46\%) compared to that without covariates (SE = 0.53\%) and is comparable to \texttt{DiD} with covariates (SE = 0.47\%). Notably, \texttt{DiD} often adopts a parametric linear model by default, such as in the \texttt{did} package, which may not efficiently handle high-cardinality categorical variables like state (with 50 levels). This can inflate variance and reduce efficiency. By contrast, our \texttt{Debiased CiC} framework leverages machine learning to flexibly and adaptively incorporate rich covariate information, improving precision without imposing restrictive functional form assumptions.

\begin{table}[ht]
\centering
\caption{Estimates of mass shootings' effect on Republican vote share (percentage points).}\label{tab:results}
\begin{tabular}{lrrr}
\toprule
Estimator & ATT (\%)& SE (\%) & 95\% CI (\%) \\
\midrule
\texttt{Debiased CiC} (with covariates, ours) & $-0.47$ & 0.46 & $[-1.37,\; \phantom{-}0.43]$ \\
\texttt{DiD} (with covariates)        & $-0.56$ & 0.47 & $[-1.49,\; \phantom{-}0.37]$ \\
\texttt{Debiased CiC} (no covariates) & $-1.63$ & 0.53 & $[-2.66,\; -0.59]$ \\
\texttt{DiD} (no covariates)         & $-2.60$ &  $0.46$ & $[-3.56,\; -1.74]$ \\
\bottomrule
\end{tabular}
\end{table}

\section{Discussion}

This paper introduces a novel extension of the Changes-in-Changes (CiC) framework that accommodates high-dimensional, non-monotonic unmeasured confounding and allows for continuous covariates. By deriving the efficient influence function and constructing Neyman-orthogonal estimators, we provide a semiparametrically efficient estimation strategy that supports valid inference with flexible, machine learning-based nuisance estimation. This work fills a key methodological gap in the Difference-in-Differences literature, where CiC has served as a foundational tool.

While we focused on the two-period setting with continuous outcomes, the proposed method can be adapted to multiple time periods, discrete outcomes, and staggered adoption designs. Future work may also extend our approach to enhance existing CiC-based methods—such as those for mediation \citep{huber2022direct} or multivariate outcomes \citep{torous2024optimal}—which often assume no covariates.

When multiple pre-treatment periods are available, the key distributional bridge assumption can, in principle, be assessed empirically, providing a basis for a falsification or placebo test. We leave the development of formal testing procedures to future work. While our simulations highlight the robustness and efficiency of the proposed estimator, careful implementation and sensitivity analyses remain essential in applied settings.

\section*{Acknowledgement}
The authors would like to acknowledge the National Institutes of Health (NIH) for their generous funding and support.

\section*{Supplementary Material} \label{SM}

Supplementary Material provides additional technical results, complete proofs, and further simulation details. The data used in the application section are publicly available from \citet{yousaf2022replication}. Code for replication is available from the first author upon request.

\bibliographystyle{plainnat}
\bibliography{mainbib}

\appendix

\renewcommand{\thefigure}{S\arabic{figure}}
\renewcommand{\thetable}{S\arabic{table}}
\renewcommand{\thetheorem}{S\arabic{theorem}}
\renewcommand{\thelemma}{S\arabic{lemma}}
\renewcommand{\theproposition}{S\arabic{proposition}}
\renewcommand{\theequation}{S\arabic{equation}}
\renewcommand{\thesection}{S\arabic{section}}

\setcounter{figure}{0}
\setcounter{table}{0}
\setcounter{theorem}{0}
\setcounter{lemma}{0}
\setcounter{equation}{0}
\setcounter{section}{0}
\setcounter{proposition}{0}

\section{Simulation Details}\label{appendix:sim}

We provide complete specifications of the data-generating process used in the simulation study.

\paragraph{Treatment Assignment.}  
To generate the linear predictor \(\phi\) in Equation~\eqref{eq:simu-1}, we set \(c_1 = -0.7\) and \(c_2 = -1.5\). These values are selected to yield a reasonable proportion of treated units and to ensure a realistic balance between the influence of measured and unmeasured confounders.

\paragraph{Potential Outcomes.}  
For the untreated potential outcomes in Equation~\eqref{eq:simu-2}, we set the scale parameters \(\alpha_0 = 3\) and \(\alpha_1 = 1\). The functional forms for the unmeasured confounding component and the time-specific covariate effects are given by:
\begin{equation*}
\begin{aligned}
  m(U, L) &= -2\left(U_1^2 + U_2 + \frac{1}{2}U_1U_2 - 1\right), \\
  k_0(L) &= -\sin\left(4\pi L_1 L_2\right) + (L_3 - 0.5)^2 + |L_4| + L_3L_5 + L_6^2 - 1, \\
  k_1(L) &= k_0(L) - 1.5 L_1 \cos(\pi L_4).
\end{aligned}
\end{equation*}
These expressions incorporate nonlinear interactions and higher-order terms to reflect realistic complexity in treatment-free outcome dynamics.

\section{Proofs}\label{appendix:proof}

\subsection{Lemma \ref{lem:QQ}}

\begin{proof}
    First, note that by Assumption \ref{assum:confounding} and \ref{assum:pos_U}, Assumption \ref{assum:db} implies
    \[\gamma(Y_0^{a = 0}, L)\mid A = 1, L, U \overset{d}{=} Y_1^{a=0} \mid A = 1, L, U.\]
    and 
    \[\gamma(Y_0^{a = 0}, L)\mid A = 0, L, U \overset{d}{=} Y_1^{a=0} \mid A = 0, L, U.\]
  We then apply the theory of optimal transport. By the Brenier-McCann's Theorem \citep{brenier1991polar, mccann1995existence}, Assumption \ref{assum:cc}, and Assumption \ref{assum:ct}(a), $Q_{Y_{1}\mid A=0, L,U}\circ F_{Y_{0}\mid A=0, L,U}\left( \cdot\right)$ is the unique function $f$ monotone in its first argument such that $f(Y_0^{a = 0}, L, U)\mid A = 0, L,U \dequiv  Y_1^{a=0} \mid A = 0,L,U \text{ almost surely.}$ Then, the existence condition by Assumption \ref{assum:db} implies $Q_{Y_{1}\mid A=0, L,U}\circ F_{Y_{0}\mid A=0, L,U}\left( \cdot\right)$ is $\gamma$ and thus 
  does not depend on $U$, almost surely.
\end{proof}

\subsection{Remark \ref{rmk:pt}}

\begin{proof}
    {(I)} We begin by showing that Assumption~\ref{assum:db} implies Equation~\eqref{eq:nlpt}. The argument closely follows the proof of Theorem~\ref{thm:idett}, with minor modifications.
    \begin{equation*}
    \begin{aligned}
        F_{Y_0^{a=0}\mid A = a, L = l}(y) 
        &= \Prob(Y_0^{a = 0} \le y \mid  A = a, L = l)\\
        & \text{(by Assumption \ref{assum:pos_U})}\\ 
        &= \int \Prob(Y_0^{a = 0} \le y \mid  A = a, L = l, U = u) \dx{F_{U\mid A = a, L = l}(u)}\\
        & \text{(by Assumption \ref{assum:ct}(b))}\\ 
        &= \int \Prob(\tgamma(Y_0^{a = 0}, l, u) \le \tgamma(y,l,u) \mid  A = a, L = l, U = u) \dx{F_{U\mid A = a, L = l}(u)}\\
        & \text{(by $\tgamma(Y_0^{a=0}, L, U)\mid A = a, L = l, U = u \sim Y_1^{a = 0}\mid A = a, L = l, U = u$)}\\
        &= \int \Prob(Y_1^{a = 0} \le \tgamma(y,l, u) \mid  A = a, L = l, U = u) \dx{F_{U\mid A = a, L = l}(u)}\\
        &= F_{Y_1^{a = 0}\mid A = a, L = l}(\gamma(y,l)).
    \end{aligned}
    \end{equation*}
    Therefore, 
    \[\gamma(y,l) = Q_{Y_1^{a = 0}\mid A = a, L = l} \circ F_{Y_0^{a=0}\mid A = a, L = l}(y),\]
    for both $a = 0, 1$.

    {(II)} Next, we show that Assumption 3.1, together with the accompanying assumptions in \citet{athey2006identification}, implies Equation~\eqref{eq:nlpt}.
    Assumption 3.1 requires
    \[
    Y_{t}^{a=0}=h\left( U,t\right) 
    \]%
    where $h$ is strictly increasing in $U$ for each $t=0,1$.
    Then we have that 
    \begin{eqnarray*}
    F_{Y_{t}^{a=0}|A=a}\left( y\right)  &=&\Pr \left\{ Y_{t}^{a=0}\leq
    y|A=a\right\}  \\
    &=&\Pr \left\{ h\left( U,t\right) \leq y|A=a\right\}  \\
    &=&\Pr \left\{ U\leq h^{-1}\left( y,t\right) |A=a\right\}  \\
    &=&F_{U|A=a}\left( h^{-1}\left( y,t\right) \right) 
    \end{eqnarray*}%
    and 
    \[
    Q_{Y_{t}^{a=0}|A=a}\left( v\right) =h\left( Q_{U|A=a}\left(
    v\right) ,t\right) 
    \]%
    therefore 
    \begin{eqnarray*}
    &&Q_{Y_{t=1}^{a=0}|A=a}\circ F_{Y_{t=0}^{a=0}|A=a}\left( y\right)  \\
    &=&h\left( Q_{U|A=a}\left( F_{U|A=a}\left( h^{-1}\left( y,t=0\right)
    \right) \right) ,t=1\right)  \\
    &=&h\left( h^{-1}\left( y,t=0\right) ,t=1\right) 
    \end{eqnarray*}%
    for both $a = 0, 1$.

    {(III)} Finally, we show that Equation~\eqref{eq:nlpt}, combined with assumptions that do not involve the unmeasured confounder \(U\)—namely, Assumptions~\ref{assum:cc} and~\ref{assum:na}—suffices to identify the ATT, as given in Equation~\eqref{eq:idett} of Theorem~\ref{thm:idett}. To establish this result, we additionally impose the following weaker continuity condition, which also avoids referencing \(U\): we assume that, almost surely, the conditional distribution functions \(F_{Y_{0} \mid A=0, L}(\cdot)\) and \(F_{Y_{1} \mid A=0, L}(\cdot)\) are continuous. Under these conditions, we have:
    \begin{equation*}
        \begin{aligned}
            &\E\left\{ Y_1 - Q_{Y_1 \mid A=0, L} \circ F_{Y_0 \mid A=0, L}(Y_0) \mid A=1 \right\}\\
            & \phantom{=}\text{  (by Assumption \ref{assum:cc} and \ref{assum:na})}\\
            &= \E[Y_1^{a=1} \mid A = 1] - \E\left\{ Q_{Y_1^{a=0} \mid A=0, L} \circ F_{Y_0^{a=0} \mid A=0, L}(Y_0^{a=0}) \mid A=1 \right\}\\
            & \phantom{=}\text{  (by Equation \eqref{eq:nlpt})}\\
            &= \E[Y_1^{a=1} \mid A = 1] - \E\left\{ Q_{Y_1^{a=0} \mid A=1, L} \circ F_{Y_0^{a=0} \mid A=1, L}(Y_0^{a=0}) \mid A=1 \right\}\\
            & \phantom{=}\text{  (by Continuity condition stated above and Law of Iterated Expectation)}\\
            &= \E[Y_1^{a=1} \mid A = 1] - \E\left\{ Y_1^{a=0} \mid A=1 \right\}\\
            &= \E[Y_1^{a=1} - Y_1^{a=0} \mid A = 1] = \theta,
        \end{aligned}
    \end{equation*}
    proving the identification. The analogous results for Corollary~\ref{cor:id_other} and for Assumption 3.1 of \citet{athey2006identification} follow by similar arguments and are omitted for brevity.
\end{proof}

\subsection{Theorem \ref{thm:idett}}

\begin{proof}
  Let ``$\sim$'' denote ``distributed as'', as a shorthand for ``$\overset{d}{=}$''. 
  By the Probability Integral Transform Theorem and Quantile Transform Theorem \citep{angus1994probability}, we have
  \begin{equation*}
    \begin{aligned}
        &Y_1^{a = 0} \mid A = 1, L,U \\
        & \text{  (by Assumptions \ref{assum:confounding} and \ref{assum:pos_U})}\\
        &\sim Y_1^{a = 0} \mid A = 0, L,U \\
        & \text{  (by transform theorems and Assumption \ref{assum:ct0})}\\
        & \sim Q_{Y_1^{a = 0} \mid A = 0, L,U} \circ F_{Y_0^{a = 0} \mid A = 0, L,U} \left( Y_0^{a = 0} \right) \mid A = 0, L,U\\
        & \text{(by Assumptions \ref{assum:confounding} and \ref{assum:na})}\\
        & \sim Q_{Y_1^{a = 0} \mid A = 0, L,U}\circ F_{Y_0^{a = 0} \mid A = 0, L,U} \left( Y_0^{a = 1} \right) \mid A = 1, L,U\\
        & \text{(by Assumption \ref{assum:cc})}\\
        & \sim Q_{Y_1 \mid A = 0, L,U} \circ F_{Y_0 \mid A = 0, L,U} \left( Y_0 \right) \mid A = 1, L,U\\
    \end{aligned}
\end{equation*}
Denote $\tgamma(y, l, u) \equiv Q_{Y_1 \mid  A = 0, L = l, U = u} \circ F_{Y_0 \mid  A = 0, L = l, U = u} \left( y \right)$.
By Assumption \ref{assum:db} and Lemma \ref{lem:QQ}, there exists a function $\gamma(y,l)$ such that $\gamma(y,l) = \tgamma(y, l, u)$. Therefore,
\begin{equation*}
  \begin{aligned}
    F_{Y_0\mid A = 0, L = l}(y) &= F_{Y_0^{a = 0}\mid A = 0, L = l}(y)\\
    &= \Prob(Y_0^{a = 0} \le y \mid  A = 0, L = l)\\
    & \text{(by Assumption \ref{assum:pos_U})}\\ 
    &= \int \Prob(Y_0^{a = 0} \le y \mid  A = 0, L = l, U = u) \dx{F_{U\mid A = 0, L = l}(u)}\\
    & \text{(by Assumption \ref{assum:ct}(b))}\\ 
    &= \int \Prob(\tgamma(Y_0^{a = 0}, l, u) \le \tgamma(y,l,u) \mid  A = 0, L = l, U = u) \dx{F_{U\mid A = 0, L = l}(u)}\\
    & \text{(by $\tgamma(Y_0^{a=0}, L, U)\mid A = 0, L = l, U = u \sim Y_1^{a = 0}\mid A = 0, L = l, U = u$)}\\
    &= \int \Prob(Y_1^{a = 0} \le \tgamma(y,l, u) \mid  A = 0, L = l, U = u) \dx{F_{U\mid A = 0, L = l}(u)}\\
    &= \int \Prob(Y_1 \le \gamma(y,l) \mid  A = 0, L = l, U = u) \dx{F_{U\mid A = 0, L = l}(u)}\\
    &= F_{Y_1\mid A = 0, L = l}(\gamma(y,l)).
  \end{aligned}
\end{equation*}
Therefore, when $\gamma(y,l)$ is in the support of the conditional law of $Y_1$ given $A = 0, L = l$, we have
\[\gamma(y,l) = Q_{Y_1\mid A = 0, L = l} \circ  F_{Y_0\mid A = 0, L = l}(y),\]
and since $\gamma(Y_0, L) \in \supp(Y_1\mid A = 0, L)$ almost surely, we have
\[ Y_1^{a = 0} \mid  A = 1, L, U \sim Q_{Y_1\mid A = 0, L} \circ  F_{Y_0\mid A = 0, L}(Y_0)\mid  A = 1, L, U.\]
Then,
\begin{equation*}
    \begin{aligned}
        \theta &= \E[Y_1^{a=1} - Y_1^{a = 0}\mid A = 1]\\
        &= \E\left[\E\left\{Y_1 - Y_1^{a = 0}\mid A = 1, L, U\right\}\right]\\
        &= \E\left[\E\left\{Y_1 - Q_{Y_1\mid A = 0, L} \circ  F_{Y_0\mid A = 0, L}(Y_0)\mid A = 1, L, U\right\}\right]\\
        &= \E\left\{Y_1 - Q_{Y_1\mid A = 0, L} \circ  F_{Y_0\mid A = 0, L}(Y_0)\mid A = 1\right\}
    \end{aligned}
\end{equation*}
\end{proof}

\subsection{Corollary \ref{cor:id_other}}

\begin{proof}
  (I) Corollary \ref{cor:idcd}:\\
  By the proof of Theorem \ref{thm:idett}, 
  \[ Y_1^{a = 0} \mid  A = 1, L, U \sim Q_{Y_1\mid A = 0, L} \circ  F_{Y_0\mid A = 0, L}(Y_0)\mid  A = 1, L, U.\] Using this result, we have 
  \begin{equation*}
    \begin{aligned}
      F_{Y^{a=0}_1\mid A = 1}(y) &= \Prob(Y^{a=0}_1 \le y \mid A = 1)\\
      &= \E\{ \Prob(Y^{a=0}_1 \le y \mid A = 1, L, U) \mid A = 1\}\\
      &= \E[ \Prob\{Q_{Y_1\mid A = 0, L} \circ  F_{Y_0\mid A = 0, L}(Y_0) \le y \mid A = 1, L, U\}  \mid A = 1]\\
      &= \Prob\{Q_{Y_1\mid A = 0, L} \circ  F_{Y_0\mid A = 0, L}(Y_0) \le y \mid A = 1\}.
    \end{aligned}
  \end{equation*} 

  (II) Corollary \ref{cor:idq}:\\
  It is a direct result of the definition of quantile treatment effect on the treated and Corollary \ref{cor:idcd}.
\end{proof}

\subsection{Proposition \ref{prop:semitrans1}}
\begin{proposition}
  Example \ref{eg:stm} satisfies Assumptions \ref{assum:setting}, \ref{assum:ct}, and \ref{assum:db}.
  \label{prop:semitrans1}
\end{proposition}

\begin{proof}
  (I) Assumption \ref{assum:confounding}: This is a direct result of $\delta \indep \epsilon_t\mid L,U$. 

  (II) Assumption \ref{assum:ct}: This is implied by the continuity of $F_{\epsilon\mid A = 0,L, U}(\cdot)$.

  (III) Assumption \ref{assum:db}: 
  Since $\beta_t$ is strictly increasing, it has inverse function $\alpha_t \equiv \beta_t^{-1}$.
  By the data-generating process in Example \ref{eg:stm}, we have the following conditional distribution function:
  \begin{equation*}
      \begin{aligned}
          F_{Y_{t}^{a=0}|A=0,L,U}\left( y\right)  &= \Prob \left( Y_{t}^{a=0}\le y|A=0,L,U\right)  \\
          &= \Prob \left\{ \alpha_{t}(Y_{t}^{a=0}) \le \alpha_{t}(y)|A=0,L,U\right\}  \\
          &= \Prob \left\{ m\left( U, L\right) +\epsilon_{t}+k_{t}(L) \le \alpha_{t}\left( y\right) |A=0,L,U\right\}  \\
          &= \Prob \left\{ \epsilon_{t} \le \alpha_{t}\left( y\right) - m\left( U, L\right) - k_{t}(L)  |A=0,L,U\right\}  \\
          &= F_{\epsilon |A=0,L,U}\left\{\alpha_{t}\left( y\right) - m\left( U, L\right) - k_{t}(L) \right\}\\
          &= F_{\epsilon |A=0,L,U}\left\{\beta_{t}^{-1}\left( y\right) - m\left( U, L\right) - k_{t}(L) \right\}
      \end{aligned}
  \end{equation*}
  This implies \[Q_{Y_{1}|A=0,L, U}\circ F_{Y_{0}|A=0, L,U}\left( y\right)  = \beta_{1}\left( \beta_{0}^{-1}\left(y\right) +k_{1}(L)-k_{0}(L)\right),\]
  which does not depend on $U$.
\end{proof}%

\subsection{Lemma \ref{lem:flat}}
For a nondecreasing real-valued function $g$ defined on $\R$, define \(I_t\) to be the largest open interval such that \(g(x) = t\) for all \(x \in I_t\). We refer to any nonempty such interval \(I_t\) as a \emph{flat spot} of \(g\), representing regions where the function \(g\) is constant.

\begin{lemma} \label{lem:flat}
Let \((\Omega, \mathcal{F}, \mathbb{P})\) be a complete probability space, and let \(X\) be a real-valued random variable. Suppose \(g: \mathbb{R} \rightarrow \mathbb{R}\) is nondecreasing, and \(X\) lies within the closure of flat spots of \(g\) with probability zero. Then, for any integrable random variable \(A\),
\[
\mathbb{E}[A \mid X] = \mathbb{E}[A \mid g(X)] \quad \text{almost surely.}
\]
\end{lemma}

\begin{proof}
Let \(\mathcal{N}\) denote the collection of all \(\mathbb{P}\)-null sets in \(\mathcal{F}\), and let \(\sigma(\mathcal{N})\) denote the \(\sigma\)-algebra they generate. Define \(Z \equiv g(X)\). Because $\sigma(A)\indep \sigma(\N)|\sigma(X), \sigma(A)\indep \sigma(\N)|\sigma(Z)$, we have
\[
\mathbb{E}[A \mid \sigma(X) \vee \sigma(\mathcal{N})] = \mathbb{E}[A \mid X], \quad \mathbb{E}[A \mid \sigma(Z) \vee \sigma(\mathcal{N})] = \mathbb{E}[A \mid Z].
\]

Since \(Z = g(X)\), it follows that \(\sigma(Z) \subseteq \sigma(X)\), and thus \(\sigma(Z) \vee \sigma(\mathcal{N}) \subseteq \sigma(X) \vee \sigma(\mathcal{N})\). We now show the reverse inclusion under the given assumption.

Define the set
\[
\mathfrak{D} \equiv \left\{ \omega \in \Omega : X(\omega) \text{ lies outside the closure of flat spots of } g \right\}.
\]
By assumption, \(\mathbb{P}(\mathfrak{D}) = 1\).

Now consider a generating class for \(\sigma(X)\), namely sets of the form \(B_t = \{x \in \mathbb{R} : x \le t\}\). For any \(t \in \mathbb{R}\),
\[
X^{-1}(B_t) = \left\{ \omega \in \Omega : X(\omega) \le t \right\}.
\]
On the set \(\mathfrak{D}\), where \(g\) is strictly increasing in a neighborhood of \(X(\omega)\), we have 
\[
\left\{ \omega \in \mathfrak{D} : X(\omega) \le t \right\} = \left\{ \omega \in \mathfrak{D} : Z(\omega) \le g(t) \right\}.
\]
Therefore,
\[
X^{-1}(B_t) = \left( \mathfrak{D} \cap \{Z \le g(t)\} \right) \cup \left( \mathfrak{D}^c \cap \{X \le t\} \right).
\]
The first term belongs to \(\sigma(Z) \vee \sigma(\mathcal{N})\), and the second to \(\sigma(\mathcal{N})\) since \(\mathbb{P}(\mathfrak{D}^c) = 0\). Thus, \(X^{-1}(B_t) \in \sigma(Z) \vee \sigma(\mathcal{N})\) for all \(t\), which implies \(\sigma(X) \subseteq \sigma(Z) \vee \sigma(\mathcal{N})\).

Hence,
\[
\sigma(X) \vee \sigma(\mathcal{N}) = \sigma(Z) \vee \sigma(\mathcal{N}),
\]
and the conditional expectations agree almost surely:
\[
\mathbb{E}[A \mid X] = \mathbb{E}[A \mid Z].
\]
\end{proof}

\subsection{Theorem \ref{thm:eif}}

\begin{proof}
  We derive the efficient influence function of the functional of observed data distribution defined by the identification formula in Equation \eqref{eq:idett}.\\
  (I) We first characterize the tangent space and its closure.
  
  Suppose the observable data $W \equiv (A, L, Y_0, Y_1) \sim P$, and define a regular parametric submodel $\mathcal{P}_\lam \equiv \{P_\lam: \lam \in (-\tilde{\lam},\tilde{\lam})\}$, with $\tilde{\lam}$ a small positive number, such that the true DGP $P = P_0$ and
  \[\log \dd{P_\lam}{P}(W) = \lam \score(W) - a(\lam),\]
  where $\score(W) \in \R$ is an arbitrary mean-zero and bounded (score) function under $P_0$, and $ a(\lam)$ is a normalizing constant with $a(\lam) = \log\E_{P_0}\left[ \exp\left\{ \lam \score(W) \right\} \right]$, $\dot{a}(\lam) = \E_{P_\lam}\left\{\score(W)\right\}$, and $\ddot{a}(\lam) = \Var_{P_\lam}\left\{\score(W)\right\}$. Evaluating at $\lam = 0$, $a(0) = \dot{a}(0) = 0, \ddot{a}(0) = \E_{P_0}\{\score(W)^2\}$. The true DGP $P_0$ satisfies Assumptions \ref{assum:setting} -- \ref{assum:db}; in fact, the whole path of $P_\lam$ for each $\lam$ and each $\score(W)$ defined above satisfies these assumptions by the boundedness of $\score$ and mutual absolute continuity between $P_\lam$ to $P_0$. This can be seen from the following:
  
  (i) Assumption \ref{assum:setting}: First, note that Assumption \ref{assum:cc}, \ref{assum:confounding}, and \ref{assum:na} do not restrict observable data. What remains is Assumption \ref{assum:pos_U}, as positivity conditional on $L,U$ implies positivity conditional on $L$ alone. 
  
  First, by the equivalent relation between $P_0$ and $P_\lam$, and boundedness of $\score$ and $\lam$, we have 
  \[m < \dd{P_\lam}{P_0} < M.\]
  Also, by Assumption \ref{assum:pos_U}, we have 
  \[P(A = 1\mid L) \in (\delta,1 - \delta). \]
  The above two bounds imply that there exists positive number $\eta$ such that for all $\score, \lam$, 
  \[P_\lam(A = 1\mid L) \in (\eta,1 - \eta).\]

  (ii) Assumption \ref{assum:ct}: this is a direct result of continuity under $P_0$ and equivalence of (conditional) measures. Note that $F_{Y_{0}\mid A=0, L,U}\left( \cdot\right)$ being a continuous function implies $F_{Y_{0}\mid A=0, L}\left( \cdot\right)$ is a continuous function. By the boundedness conditions above, we know that the conditional density ratios are also bounded. Therefore, for a measurable set $B$, whenever conditional measure of $\{Y_0 \in B\}$ is zero given $A = 0, L$, $P_\lam(Y_0 \in B \mid A = 0, L)$ is also zero. 
  
  (iii) Assumption \ref{assum:db}: This condition does not restrict observable data.
  
  Therefore, since bounded measurable functions are dense in $L_2(P_0)$, the mean-square closure of the linear span of all $\score$ defined above is the entire Hilbert space of $P_0$-mean-zero measurable functions of $W$ with finite second moments equipped with the inner product $\langle h_1, h_2 \rangle = \E_{P_0}\left\{ h_1(W)h_2(W) \right\}$. In other words, our statistical model is nonparametric, and thus there is at most one influence function, and if there exists one, it is the efficient influence function.\\
  
  (II) We derive below the efficient influence function (EIF) for $\theta = \theta_1 - \theta_2$ as a pathwise derivative, where
  \begin{equation*}
    \begin{aligned}
      \theta_1 &= \E\left\{ Y_{1}\mid A=1\right\},\\
      \theta_2 &= \E\left\{ Q_{Y_{1}\mid A=0, L}\circ F_{Y_{0}\mid A=0, L}\left( Y_{0}\right)\mid A=1\right\}. 
    \end{aligned}
  \end{equation*}
  The EIF for $\theta$, denoted by $\IF{\theta}$, is the difference of the EIFs for $\theta_1$ and $\theta_2$, i.e. $\IF{\theta_1}$ and $\IF{\theta_2}$.
  It is well known that 
  \[\IF{\theta_1} = \ff{\I{A = 1}}{P(A = 1)}(Y_1 - \theta_1) = \ff{A}{P(A = 1)}(Y_1 - \theta_1),\]
  where $\I{\cdot}$ is the indicator function. Next, we focus on $\IF{\theta_2}$.

  For ease of notation, let 
  $h(y,l) \equiv Q_{Y_{1}\mid A=0,L = l}(y)$, $ g(y,l) \equiv F_{Y_{0}\mid A=0,L = l}(y) $, and $\gamma(y,l) = h(g(y,l),l)$. Then, under $P_\lam$, we have $\gamma_\lam(y,l) = h_\lam(g_\lam(y,l),l) \equiv Q_{\lam,Y_{1}\mid A=0, L = l} \circ  F_{\lam, Y_{0}\mid A=0, L = l}(y)$, and
  \begin{equation*}
    \begin{aligned}
      \theta_{2, \lam} &= \E_{P_\lam}\left\{ \gamma_\lam(Y_0, L) \mid  A = 1\right\} \\
      &= \int \gamma_\lam(y_0, l) f_\lam(y_0, l\mid A = 1) \dy{y_0}\dy{l} \\
      &= \int \gamma_\lam(y_0, l) \ff{f_\lam(y_0, l, A = 1)}{P_\lam(A = 1)} \dy{y_0}\dy{l}.
    \end{aligned}
  \end{equation*}
  Let $\score(\cdot)$ denote the score function. Then, the pathwise derivative is
  \begin{equation*}
    \begin{aligned}
      \plam\theta_{2, \lam}\midl &= 
      \int \pp{}{\lam}\gamma_\lam(y_0, l)\bmidl \ff{f(y_0, l, A = 1)}{P(A = 1)} \dy{y_0}\dy{l} 
      + \int \gamma(y_0,l) \pp{}{\lam}\left[\ff{f_\lam(y_0,l, A = 1)}{P_\lam(A = 1)} \right]\bmidl\dy{y_0}\dy{l}\\
      &= \E\left[\pp{}{\lam}\gamma_\lam(Y_0,L)\bmidl \bigg| A = 1\right] + \E[\score(Y_0, L|A = 1)\gamma(Y_0,L)|A = 1]
    \end{aligned}
  \end{equation*}
  Next, we will focus on $\E\left[\pp{}{\lam}\gamma_\lam(Y_0, L)\bmidl \bigg| A = 1\right]$. By Assumption \ref{assum:confounding} and \ref{assum:pos_U} and measure changing, we have
  \begin{equation*}
    \begin{aligned}
      \E\left[\pp{}{\lam}\gamma_\lam(Y_0, L)\bmidl \bigg| A = 1\right] &= \ff{P(A = 0)}{P(A = 1)}\E\left[\plam\gamma_\lam(Y_0, L)\cdot \ff{P(A = 1|Y_0, L)}{P(A = 0|Y_0, L)} \bigg| A = 0\right]\\
      &= \ff{P(A = 0)}{P(A = 1)}\E\left[\E\left\{\plam\gamma_\lam(Y_0, L)\cdot \ff{P(A = 1|Y_0, L)}{P(A = 0|Y_0, L)}\bigg| A = 0, L\right\} \bigg| A = 0\right]
    \end{aligned}
  \end{equation*}
  Since under $P_\lam$, we have
  \[\gamma_\lam(Y_0, L)|A = 0,L \sim Y_1|A = 0, L.\]
  Therefore, for an arbitrary measurable function $\rho(x,l)$, we have the identity
  \[0 = \E_\lam[\rho(\gamma_\lam(Y_0, L), L) - \rho(Y_1, L)|A = 0, L].\]
  Taking partial derivative with respect to $\lam$ on both sides, we have
  \begin{equation}
  0 = \E[\score(Y_1, Y_0|A = 0, L)\{\rho(\gamma(Y_0,L), L) - \rho(Y_1, L)\}|A = 0, L] + \E[\partial_x\rho(\gamma(Y_0, L), L)\cdot \plam \gamma_\lam(Y_0, L)|A = 0, L].
  \label{eq:eqdiff}
  \end{equation}
  Then, given $L = l$, for an arbitrary random variable $X$, we have \[\E[X|\gamma(Y_0, l), L = l] = \E[X|F_{Y_0|A = 0, L = l}(Y_0), L = l] = \E[X|Y_0, L = l],\]
  where the first equality is by the strict monotonicity of $Q_{Y_1\mid A = 0, L = l}$, and the second equality is by Lemma \ref{lem:flat}.

  Then, by choosing an $\rho$ whose partial derivative with respect to its first argument $x$ is
  \[\nu(x,l) \equiv \partial_x\rho(x,l) \equiv  \ff{P(A = 1|\gamma(Y_0, l) = x, L = l)}{P(A = 0|\gamma(Y_0, l) = x, L = l)}, \] and applying Equation \eqref{eq:eqdiff}, we have
  \begin{equation*}
  \begin{aligned}
      \E\left[\pp{}{\lam}\gamma_\lam(Y_0)\bmidl \bigg| A = 1\right]
      &= \ff{P(A = 0)}{P(A = 1)}\E\left[\E\left\{\plam\gamma_\lam(Y_0, L)\cdot \ff{P(A = 1|Y_0, L)}{P(A = 0|Y_0, L)}\bigg| A = 0, L\right\} \bigg| A = 0\right]\\
      &= \ff{P(A = 0)}{P(A = 1)}\E\left[-\E[\score(Y_1, Y_0|A = 0, L)\{\rho(\gamma(Y_0,L), L) - \rho(Y_1, L)\}|A = 0, L]\bigg| A = 0\right]\\ 
      &= -\ff{P(A = 0)}{P(A = 1)}\E\left[\E\left\{\score(Y_1, Y_0|A = 0, L)\left(\int_{Y_1}^{\gamma(Y_0, L)} \nu(x,L) \dy{x}\right)\bigg|A = 0, L\right\}\bigg| A = 0\right]\\
      &= -\ff{P(A = 0)}{P(A = 1)}\E\left[\score(Y_1, Y_0|A = 0, L)\left(\int_{Y_1}^{\gamma(Y_0, L)} \nu(x,L) \dy{x}\right)\bigg| A = 0\right]\\
      &= \E\left[\score(Y_1, Y_0|A = 0, L)\left(\ff{A - 1}{P(A = 1)}\int_{Y_1}^{\gamma(Y_0, L)} \nu(x,L) \dy{x}\right)\right].
  \end{aligned}
  \end{equation*}
  Therefore,
  \begin{equation*}
    \begin{aligned}
      \plam\theta_{2, \lam}\midl 
      &= \E\left[\score(Y_1, Y_0|A = 0, L)\left(\ff{A - 1}{P(A = 1)}\int_{Y_1}^{\gamma(Y_0, L)} \nu(x,L) \dy{x}\right)\right] + \E\left[\score(Y_0, L|A = 1)\gamma(Y_0,L)\ff{A}{P(A = 1)}\right]
    \end{aligned}
  \end{equation*}
  Then, we have 
  \begin{equation*}
  \begin{aligned}
      & \E\left[\score(Y_1, Y_0|A = 0, L)\left(\ff{A - 1}{P(A = 1)}\int_{Y_1}^{\gamma(Y_0, L)} \nu(x,L) \dy{x}\right)\right] \\
      &= -\ff{P(A = 0)}{P(A = 1)}\E\left[\score(Y_1, Y_0|A, L)\left(\int_{Y_1}^{\gamma(Y_0, L)} \nu(x,L) \dy{x}\right)\bigg| A = 0\right]\\
      &= -\ff{P(A = 0)}{P(A = 1)}\E\left[(\score(Y_1, Y_0|A, L) + \score(A, L))\left(\int_{Y_1}^{\gamma(Y_0, L)} \nu(x,L) \dy{x}\right)\bigg| A = 0\right]\\
      &= -\ff{P(A = 0)}{P(A = 1)}\E\left[\score(Y_1, Y_0, A, L)\left(\int_{Y_1}^{\gamma(Y_0, L)} \nu(x,L) \dy{x}\right)\bigg| A = 0\right]\\
      &= \E\left[\score(Y_1, Y_0, A, L)\left(\int_{Y_1}^{\gamma(Y_0, L)} \nu(x,L) \dy{x}\right)\ff{A - 1}{P(A = 1)}\right]
  \end{aligned}
  \end{equation*}

  \begin{equation*}
    \begin{aligned}
      & \E[\score(Y_0, L|A = 1)\gamma(Y_0, L)|A = 1] \\
      &= \E\left[\score(Y_0, L|A = 1)\gamma(Y_0, L)\ff{A}{P(A = 1)}\right] \\
      &= \E\left[\score(Y_0, L|A)\gamma(Y_0, L)\ff{A}{P(A = 1)}\right] \\
      &= \E\left[\score(Y_0, L|A)\{\gamma(Y_0, L) - \E[\gamma(Y_0, L)|A]\}\ff{A}{P(A = 1)}\right] \\
      &= \E\left[\left\{\score(Y_0, L|A) + \score(Y_1|Y_0,A, L)\right\}\{\gamma(Y_0, L) - \E[\gamma(Y_0, L)|A]\}\ff{A}{P(A = 1)}\right] \\
      &= \E\left[\left\{\score(Y_1, Y_0, L|A) + \score(A)\right\}\{\gamma(Y_0, L) - \E[\gamma(Y_0, L)|A]\}\ff{A}{P(A = 1)}\right] \\
      &= \E\left[\score(Y_1, Y_0, L, A) \{\gamma(Y_0, L) - \E[\gamma(Y_0, L)|A]\}\ff{A}{P(A = 1)}\right] \\
      &= \E\left[\score(Y_1, Y_0, L, A) \{\gamma(Y_0, L) - \E[\gamma(Y_0, L)|A = 1]\}\ff{A}{P(A = 1)}\right]
    \end{aligned}
  \end{equation*}
  Then,
  \begin{equation*}
    \begin{aligned}
      \plam\theta_{2, \lam}\midl
      &= \E\left[\score(Y_1, Y_0, L, A) \left(\{\gamma(Y_0, L) - \E[\gamma(Y_0, L)|A = 1]\}\ff{A}{P(A = 1)} + \int_{Y_1}^{\gamma(Y_0, L)} \nu(x,L) \dy{x}\ff{A - 1}{P(A = 1)}\right)\right]\\
      &= \E\left[\ff{\score(Y_1, Y_0, L, A) }{P(A = 1)}\left\{ A(\gamma(Y_0, L) - \theta_2) - (1-A)\int_{Y_1}^{\gamma(Y_0, L)} \nu(x, L)\dx{x}\right\}\right].
    \end{aligned}
  \end{equation*}
  Since our model is nonparametric, the EIF for $\theta_2$ is
  \[\IF{\theta_2} = \ff{1}{P(A = 1)}\left\{ A(\gamma(Y_0, L) - \theta_2) - (1-A)\int_{Y_1}^{\gamma(Y_0, L)} \nu(x, L)\dx{x}\right\},\]
  and the EIF for $\theta$ is
  \begin{equation*}
    \begin{aligned}
      \IF{\theta} = \ff{A}{P(A = 1)}[\{Y_1 - \gamma(Y_0, L)\} - \theta] + \ff{1 - A}{P(A = 1)}\int_{Y_1}^{\gamma(Y_0, L)} \nu(x, L)\dx{x}.
    \end{aligned}
  \end{equation*}

(III) Therefore, the semiparametric efficiency bound is $\Var(\IF{\theta}) = \E(\IF{\theta}^2)$.
  
\end{proof}

\subsection{Theorem \ref{thm:eif_gen}}

\begin{proof}
  (I) Tangent Space: 
  
  The mean-square closure of the tangent space remains the nonparametric model, as in the proof of Theorem \ref{thm:eif}, thus omitted. In particular, we still adopt the same set of parametric submodels: 
  \[\log \dd{P_\lam}{P}(W) = \lam \score(W) - a(\lam).\]

  (II) Efficient Influence Function (EIF):

  Since $\E_{P_0}[g(W;\vartheta, \gamma) \mid A = 1] = 0$, in a submodel $\mathcal{P}_\lam \equiv \{P_\lam: \lam \in (-\tilde{\lam},\tilde{\lam})\}$, we have the identity $0 =\E_{\lam}[g(W;\vartheta_\lam, \gamma_\lam)\mid A = 1]$. In this proof, all the partial derivative with respect to $\lam$ is evaluated at $\lam = 0$ and we thus omit this for notational convenience. Taking derivative on both sides, we have
  \begin{equation*}
    \begin{aligned}
      0 &= \pd{\lam}\E_{\lam}[g(W;\vartheta_\lam, \gamma_\lam)\mid A = 1]\\
      &= \E[g(W;\vartheta, \gamma)\score(W\mid A = 1)\mid A = 1] + \pd{\vartheta}\E[g(W;\vartheta, \gamma)\mid A = 1]\cdot \pd{\lam}\vartheta_\lam \\
      & \phantom{=} + \pd{\lam}\E[g(W;\vartheta, \gamma_\lam)\mid A = 1]
    \end{aligned}
  \end{equation*}
  Therefore, 
  \[\pd{\lam}\vartheta_\lam = \ff{\E[g(W;\vartheta, \gamma)\score(W\mid A = 1)\mid A = 1] + \pd{\lam}\E[g(W;\vartheta, \gamma_\lam)\mid A = 1] }{-\pd{\vartheta}\E[g(W;\vartheta, \gamma)\mid A = 1]}.\]
  In the following we change the numerator into the form of a pathwise derivative.

  (i) $\pd{\lam}\E[g(W;\vartheta, \gamma_\lam)\mid A = 1]$: 

  We will use the identity again previously shown in the proof of Theorem \ref{thm:eif}:
  \begin{equation*}
    0 = \E[\score(Y_1, Y_0|A = 0, L)\{\rho(\gamma(Y_0,L), L) - \rho(Y_1, L)\}|A = 0] + \E[\partial_x\rho(\gamma(Y_0, L), L)\cdot \plam \gamma_\lam(Y_0, L)|A = 0].
  \end{equation*}
  Denote $\pi(Y_0, L) \equiv \Prob_0(A = 1 \mid Y_0, L)$, and recall that\[\nu(x,l) \equiv \ff{P(A = 1|\gamma(Y_0, l) = x, L = l)}{P(A = 0|\gamma(Y_0, l) = x, L = l)}.\] Then, we have 
  \begin{equation*}
    \begin{aligned}
      &\pd{\lam}\E[g(W;\vartheta, \gamma_\lam)\mid A = 1]\\
      &= \pd{\lam}\E[\tilde{g}(\gamma_\lam(Y_0, L),\vartheta)\mid A = 1]\\
      &= \E\left[\pd{\gamma}\tilde{g}(\gamma(Y_0, L),\vartheta) \cdot \pd{\lam}\gamma_\lam(Y_0, L) \mid A = 1\right]\\
      &= \ff{1 - \pi}{\pi}\E\left[\pd{\gamma}\tilde{g}(\gamma(Y_0, L),\vartheta) \cdot \ff{\pi(Y_0,L)}{1 - \pi(Y_0,L)} \cdot \pd{\lam}\gamma_\lam(Y_0, L) \mid A = 0\right]\\
      &= \ff{1 - \pi}{\pi}\E\left[\pd{\gamma}\tilde{g}(\gamma(Y_0, L),\vartheta) \cdot \nu\{\gamma(Y_0,L), L\} \cdot \pd{\lam}\gamma_\lam(Y_0, L) \mid A = 0\right]\\
      &= \ff{1 - \pi}{-\pi}\E\left[\score(Y_1, Y_0|A = 0, L)\int_{Y_1}^{\gamma(Y_0, L)} \pd{x}\tilde{g}(x,\vartheta) \cdot \nu(x, L) \dx{x} \bigg| A = 0\right]\\
      &= \ff{1 - \pi}{-\pi}\E\left[\score(Y_1, Y_0, A, L)\int_{Y_1}^{\gamma(Y_0, L)} \pd{x}\tilde{g}(x,\vartheta) \cdot \nu(x, L) \dx{x} \bigg| A = 0\right]\\
      &= \E\left[\score(Y_1, Y_0, A, L) \cdot \ff{A - 1}{\pi}\int_{Y_1}^{\gamma(Y_0, L)} \nu(x, L)\dx_x{\tilde{g}(x,\vartheta)} \right]
    \end{aligned}
  \end{equation*}

  (ii) $\E[g(W;\vartheta, \gamma)\score(W\mid A = 1)\mid A = 1] = \E[\ff{A}{\pi}g(W;\vartheta, \gamma)\score(W)]$.

  (III) Summary:

  Therefore, the EIF for $\vartheta$ is
  \[\IF{\vartheta} = \ff{\ff{A}{\pi}g(W;\vartheta, \gamma) + \ff{A - 1}{\pi}\int_{Y_1}^{\gamma(Y_0, L)} \nu(x, L)\dx_x{\tilde{g}(x,\vartheta)} }{-\pd{\vartheta}\E[g(W;\vartheta, \gamma)]}.\]
  And by the semiparametric theory, the semiparametric efficiency bound for $\vartheta$ is $\E\left\{\IF{\vartheta}^2\right\}$.

\end{proof}

\subsection{Corollary \ref{cor:eif_other}}

\begin{proof}
  (a) $\vartheta^{CDT,y}$:

  Recall that $g(W;\vartheta, \gamma) = \tilde{g}(\gamma(Y_0, L), \vartheta) = \I{\gamma(Y_0, L) < y}- \vartheta$, and therefore $ \pd{\vartheta}g(W;\vartheta, \gamma) = -1$, and $ \pd{x}\tilde{g}(x,\vartheta) = -\delta_y(x)$, where $\delta_y(\cdot) \equiv \delta(\cdot - y)$ and $\delta(\cdot)$ is the Dirac delta function.

  Then, the denominator $-\pd{\vartheta}\E[g(W;\vartheta, \gamma)] = 1$, and 
  \begin{equation*}
    \begin{aligned}
      &\int_{Y_1}^{\gamma(Y_0, L)} \nu(x, L)\dx_x{\tilde{g}(x,\vartheta)}\\
      &= -\int_{Y_1}^{\gamma(Y_0, L)} \nu(x, L)\delta_y(x) \dy{x}\\
      &= \nu(y, L) \chi(y, W; \gamma)
    \end{aligned}
  \end{equation*}

  Therefore, we have
  \begin{equation*}
    \begin{aligned}
      \IF{\vartheta^{CDT,y}} &= \ff{A}{\pi}\left[\I{\gamma(Y_0, L) < y}- \vartheta^{CDT,y}\right] + \ff{A - 1}{\pi} \nu(y, L) \chi(y, W; \gamma).
    \end{aligned}
  \end{equation*}

  (b) $\vartheta^{QTT,\tau}$:

  (I) First, note that the additional condition that $Q_{Y^{a}_1\mid A = 1}(\tau)$ is unique and $Y^{a}_1\mid A = 1$ does not have point mass at $Q_{Y^{a}_1\mid A = 1}(\tau)$ for $a = 0,1$ does not change the closure of the tangent space, by Assumption \ref{assum:pos_U} and \ref{assum:ct}. Therefore, we will continue to work with the nonparametric model.

  (II) The EIF of $\vartheta^{QTT,\tau}$ is the difference between $\IF{\vartheta_1}$ and $\IF{\vartheta_2}$, where the first one is well-known in the literature, e.g. \citet{firpo2007efficient}, to be 
  \[\IF{\vartheta_1} = -\ff{A}{\pi}\ff{\I{Y_1 \le \vartheta_1} - \tau}{f_{Y_1\mid A = 1}(\vartheta_1)}.\] 
  
  In the following, we focus on deriving $\IF{\vartheta_2}$ by applying Theorem \ref{thm:eif_gen}. We may choose $g(W;\vartheta_2, \gamma)  = \tilde{g}(\gamma(Y_0,L);\vartheta_2) = \I{\gamma(Y_0, L) < \vartheta_2} - \tau$. Then, we have $\pd{\vartheta_2}\E [g \mid A = 1] = f_{\gamma(Y_0,L)\mid A = 1}(\vartheta_2)$, and $\pd{x}\tilde{g}(x,\vartheta_2) = -\delta_{\vartheta_2}(x)$. Therefore, 
  \begin{equation*}
    \begin{aligned}
      &\int_{Y_1}^{\gamma(Y_0, L)} \nu(x, L)\dx_x{\tilde{g}(x,\vartheta_2)}\\
      &= -\int_{Y_1}^{\gamma(Y_0, L)} \nu(x, L)\delta_{\vartheta_2}(x) \dy{x}\\
      &=  \nu(\vartheta_2, L) \chi(\vartheta_2, W; \gamma).
    \end{aligned}
  \end{equation*}
  Therefore,
  \begin{equation*}
    \begin{aligned}
      \IF{\vartheta_2} &= \ff{\ff{A}{\pi}\left[\I{\gamma(Y_0, L) < \vartheta_2} - \tau\right] + \ff{A - 1}{\pi}\nu(\vartheta_2, L) \chi(\vartheta_2, W; \gamma)}{-f_{\gamma(Y_0,L)\mid A = 1}(\vartheta_2)}.
    \end{aligned}
  \end{equation*}
  To conclude, $\IF{\vartheta^{QTT,\tau}} = \IF{\vartheta_1} -  \IF{\vartheta_2}$.

\end{proof}

\subsection{Lemma \ref{lem:bias}}

\begin{proof}
Recall that
\[
\psi(\theta,\eta_\lambda)
= \pi_\lambda^{-1}
  \Bigl[
    A\bigl(Y_1 - \gamma_\lambda(Y_0,L) - \theta\bigr)
    + (1-A)\!\int_{Y_1}^{\gamma_\lambda(Y_0,L)}\!\!\nu_\lambda(x,L)\dx{x}
  \Bigr].
\]
Set
\[
S_\lambda
:= A\bigl(Y_1 - \gamma_\lambda(Y_0,L) - \theta\bigr)
   + (1-A)\!\int_{Y_1}^{\gamma_\lambda(Y_0,L)}\!\!\nu_\lambda(x,L)\dx{x},
\]
so that
\[
\Phi(\lambda)
:= \E\bigl[\psi(\theta,\eta_\lambda)\bigr]
= \pi_\lambda^{-1}\,\E\bigl[S_\lambda\bigr]
= \pi_\lambda^{-1}\,g(\lambda),
\quad
g(\lambda):=\E\bigl[S_\lambda\bigr].
\]

Let $u_\lambda := \gamma_\lambda(Y_0,L)$.  Then by the chain and Leibniz rules,
\begin{align*}
S_\lambda'
&= A\bigl[-\Delta\gamma(Y_0,L)\bigr]
  + (1-A)\,\dd{}{\lam}\!\int_{Y_1}^{u_\lambda}\!\!\nu_\lambda(x,L)\dx{x} \\[4pt]
&= -\,A\,\Delta\gamma(Y_0,L)
   + (1-A)\Bigl[
       \Delta\gamma(Y_0,L)\,\nu_\lambda(u_\lambda,L)
       + \int_{Y_1}^{u_\lambda}\!\!\Delta\nu(x,L)\dx{x}
     \Bigr],\\[6pt]
S_\lambda''
&= -\,A\cdot 0
   + (1-A)\,\dd{}{\lam}\Bigl[
       \Delta\gamma\,\nu_\lambda(u_\lambda,L)
       + \int_{Y_1}^{u_\lambda}\!\!\Delta\nu(x,L)\dx{x}
     \Bigr]\\
&= (1-A)\Bigl[
     2\,\Delta\gamma(Y_0,L)\,\Delta\nu\bigl(u_\lambda,L\bigr)
     + \bigl[\Delta\gamma(Y_0,L)\bigr]^2
       \,\partial_x\nu_\lambda\bigl(u_\lambda,L\bigr)
   \Bigr].
\end{align*}

Since $\pi_\lambda'=\Delta\pi$, $\pi_\lambda''=0$, we get
\begin{align*}
\Phi'(\lambda)
&= \dd{}{\lam}\Bigl[\pi_\lambda^{-1}\,g(\lambda)\Bigr]\\
&= -\,\frac{\Delta\pi}{\pi_\lambda^2}\,g(\lambda)
   + \frac{1}{\pi_\lambda}\,g'(\lambda)
= \frac{1}{\pi_\lambda}\,\E\bigl[S_\lambda'\bigr]
  \;-\;\frac{\Delta\pi}{\pi_\lambda^2}\,\E\bigl[S_\lambda\bigr].
\end{align*}
When evaluated at $\lam = 0$, we have the expression in (a):
\[\Phi'(0)  = 0,\] by the properties of efficient influence function \citep{newey1994asymptotic}. 

Furthermore,
\begin{align*}
\Phi''(\lambda)
&= \dd{}{\lam}\Bigl[\Phi'(\lambda)\Bigr]\\
&= \frac{1}{\pi_\lambda}\,g''(\lambda)
  -2\,\frac{\Delta\pi}{\pi_\lambda^2}\,g'(\lambda)
  +2\,\frac{(\Delta\pi)^2}{\pi_\lambda^3}\,g(\lambda)\\
&= \frac{1}{\pi_\lambda}\,\E\bigl[S_\lambda''\bigr]
  -2\,\frac{\Delta\pi}{\pi_\lambda^2}\,\E\bigl[S_\lambda'\bigr]
  +2\,\frac{(\Delta\pi)^2}{\pi_\lambda^3}\,\E\bigl[S_\lambda\bigr].
\end{align*}
Therefore,
\begin{equation}\label{eq:second-diff}
    \Phi''(\lambda) = \frac{1}{\pi_\lambda}\,\E\bigl[S_\lambda''\bigr]
  -2\,\frac{\Delta\pi}{\pi_\lambda^2}\,\E\bigl[S_\lambda'\bigr]
  +2\,\frac{(\Delta\pi)^2}{\pi_\lambda^3}\,\E\bigl[S_\lambda\bigr],
\end{equation}
which is a continuous function in $\lam$ by Assumption \ref{assum:pos_U} and \ref{assum:nu}. Then, the expression in (b) follows as a simplified form of the result above.
\end{proof}

\subsection{Theorem \ref{thm:asymptotic}}

\begin{proof}
The asymptotic properties stated in Theorem~\ref{thm:asymptotic} follow from Theorems 3.1 and 3.2, and Corollaries 3.1–3.3 of \citet{chernozhukov2018double}. To apply these results, we verify that the key conditions—Assumptions 3.1 and 3.2 in their framework—are satisfied in our setting.

We begin with Assumption 3.1 of \citet{chernozhukov2018double}. Conditions (a), (b), (d), and (e) are directly satisfied by standard properties of the efficient influence function; see, e.g., \citet{newey1990semiparametric, newey1994asymptotic}. Condition (c) follows from the second derivative characterization in Equation~\eqref{eq:second-diff} of the proof of Lemma~\ref{lem:bias}.

Next, we verify Assumption 3.2. Conditions (a), (b), and (c) are implied by the convergence rate assumption on the nuisance estimators stated in Assumption~\ref{assum:rate} as well as the regularity condition Assumption \ref{assum:nu}. Condition (d) is trivially satisfied.

This completes the verification of the required assumptions.
\end{proof}

\end{document}